\documentclass{article}
\usepackage{amsmath,amssymb,amsfonts}
\usepackage{graphicx}
\usepackage{textcomp}
\usepackage{float}
\usepackage[margin=1in]{geometry}
\usepackage{commath,mathtools}
\usepackage{algorithm}
\usepackage{algorithmic}
\usepackage{subfigure}
\usepackage[english]{babel}
\usepackage[autostyle]{csquotes}
\usepackage{hyperref}
\usepackage{color}
\usepackage{rotating,booktabs}
\usepackage{lscape}
\usepackage[sort,nocompress]{cite}
\usepackage{enumitem}
\usepackage{multirow}
\usepackage{ulem} 
\usepackage{tikz}
\usepackage{comment}
\usetikzlibrary{shapes,arrows,chains}
\usepackage{amsthm}
\usepackage{amssymb}
\newtheorem{theorem}{Theorem}[section]

\newtheorem{lemma}[theorem]{Lemma}
\newtheorem{definition}{Definition}[section]
\usepackage{pifont}
\newcommand{\cmark}{\ding{52}}%
\newcommand{\xmark}{\ding{56}}%
\usepackage[pagewise,mathlines]{lineno}
\DeclareFontFamily{U}{matha}{\hyphenchar\font45}
\DeclareFontShape{U}{matha}{m}{n}{
      <5> <6> <7> <8> <9> <10> gen * matha
      <10.95> matha10 <12> <14.4> <17.28> <20.74> <24.88> matha12
      }{}
\DeclareSymbolFont{matha}{U}{matha}{m}{n}

\DeclareMathSymbol{\Lt}{3}{matha}{"CE}
\DeclareMathSymbol{\Gt}{3}{matha}{"CF}

\newcommand{\bbf}{\mathbf{b}}

\newcommand{\xbf}{\mathbf{x}}

\newcommand{\zbf}{\mathbf{z}}
\newcommand{\gbf}{\mathbf{g}}
\newcommand{\ubf}{\mathbf{u}}
\newcommand{\vbf}{\mathbf{v}}
\newcommand{\wbf}{\mathbf{w}}

\newcommand{\reps}{r_{\varepsilon}}
\newcommand{\hatreps}{\hat{r}_{\varepsilon}}

\newcommand{\xk}{\mathbf{x}_k}

\newcommand{\xkp}{\mathbf{x}_{k+1}}
\newcommand{\zkp}{\mathbf{z}_{k+1}}
\newcommand{\ukp}{\mathbf{u}_{k+1}}
\newcommand{\vkp}{\mathbf{v}_{k+1}}

\newcommand{\Ccal}{\mathcal{C}}

\newcommand{\Xcal}{\mathcal{X}}

\newcommand{\gi}{\gbf_i}
\newcommand{\epsk}{\varepsilon_{k}}

\newcommand{\kl}{k_{l}}

\newcommand{\sji}{\mathbf{s}_{j,i}}

\newcommand{\xjp}{\mathbf{x}_{j+1}}
\newcommand{\epsj}{\varepsilon_{j}}
\newcommand{\xhat}{\hat{\mathbf{x}}}

\newcommand{\etol}{\epsilon_{\mathrm{tol}}}

\DeclareMathOperator{\dist}{\mathrm{dist}}
\DeclareMathOperator{\prox}{\mathrm{prox}}
\DeclareMathOperator*{\argmin}{\mathrm{arg\,min}}

\def\BibTeX{{\rm B\kern-.05em{\sc i\kern-.025em b}\kern-.08em
    T\kern-.1667em\lower.7ex\hbox{E}\kern-.125emX}}
\begin{document}
\title{Provably Convergent Learned Inexact Descent Algorithm for Low-Dose CT Reconstruction}
\author{Qingchao Zhang, Mehrdad Alvandipour,  Wenjun Xia, Yi Zhang, Xiaojing Ye and Yunmei Chen
\thanks{Q. Zhang and M. Alvandipour contributed equally. Corresponding author: Y. Chen.}
\thanks{Q. Zhang, M. Alvandipour, Y. Chen are with the Department of Mathematics, University of Florida, Gainesville, FL 32611 USA (e-mail: \{qingchaozhang, m.alvandipour, yun\}@ufl.edu).}
\thanks{W. Xia and Y. Zhang are with  with the College of
Computer Science, Sichuan University, Chengdu 610065 China (e-mail:  xwj90620@gmail.com; yzhang@scu.edu.cn).}
\thanks{X. Ye is with the Department of Mathematics and Statistics, Georgia State University, Atlanta, GA 30303 USA
(e-mail: xye@gsu.edu).}}
\maketitle
\begin{abstract}
We propose a provably convergent method, called Efficient Learned Descent Algorithm (ELDA), for low-dose CT (LDCT) reconstruction. ELDA is a highly interpretable neural network architecture with learned parameters and meanwhile retains convergence guarantee as classical optimization algorithms. To improve reconstruction quality, the proposed ELDA also employs a new non-local feature mapping and an associated regularizer. We compare ELDA with several state-of-the-art deep image methods, such as RED-CNN and Learned Primal-Dual, on a set of LDCT reconstruction problems. Numerical experiments demonstrate improvement of reconstruction quality using ELDA with merely 19 layers, suggesting promising performance of ELDA in solution accuracy and parameter efficiency.
\end{abstract}


\section{Introduction}
\label{sec:introduction}
Computed Tomography (CT) is one of the most widely used imaging technologies for medical diagnosis. CT employs X-ray measurements from different angles to generate cross-sectional images of the human body \cite{houns,cormack1964representation}. As high dosage X-rays can be harmful to human body \cite{cancer1,cancer2, cancer3}, substantial efforts have been devoted to image reconstruction using low-dose CT measurements \cite{dose-reduction1,dose-reduction2,dose-reduction3}. There are two main strategies for dose reduction in CT scans: one is to reduce the number of views, and the other is to reduce the exposure time and the current of X-ray tube \cite{hsieh2019performance}, both of which will introduce various degrees of noise and artifacts and then compromise the subsequent diagnosis. Here we focus on the second type however our method is not specific to a particular scanning mode. We formulate the problem as an optimization problem which will be solved with our Efficient Learned Descent Algorithm (ELDA).

The classic analytical method  to  reconstruct  CT images from projection data, Filtered  Back-Projection (FBP), leads  to  heavy  noise  and  artifacts  in  the  low dose scenario.
The remedy for this problem have been sought from three different perspectives: pre-processing the sinograms \cite{sino-domain1,sino-domain2,sino-domain3}, post-processing the images \cite{image_domain-tipnis2010iterative}, or the hybrid approach with iterative reconstructions that encode prior information into the process \cite{iterative-1,iterative-2,zheng2018pwls}. 



The advent of machine learning methods and its success on various image processing  applications, have naturally led to incorporation of deep models into all of the above  approaches and produced a better performance than analytical methods \cite{shan2019competitive}. For instance,  CNN  methods\cite{CNN1,CNN2,CNN3, xie2020artifact}, that have been applied to sparse view \cite{zhang2018sparse, hu2020hybrid, jin2017deep, xie2020artifact, lee2018deep} and low dose \cite{CNN1,CNN2,CNN3,kang2018deep, yang2018low} data.
It is also applied in projection domain synthesis \cite{lee2018deep, proj_domain-lee2017view}, post processing \cite{CNN4, jin2017deep, CNN3, kang2017deep, xie2020artifact}, and for prior learning in iterative methods \cite{ye2019spultra, chen2018learn, wu2017iterative}.



Recently, a number of learned optimization methods have been proposed and are proven very effective in CT reconstruction problem, as they are able to learn adaptive regularizer which leads to more accurate image reconstruction in a variety of medical imaging applications.
However, existing works model regularizers using convolutional neural networks (CNNs) which only explore local image features. This limits the representation power of deep neural networks and is not suitable for medical imaging applications which demand high image qualities. Moreover, most of existing deep networks for image reconstruction are cast as black-boxes and can be difficult to interpret. Last but not least, deep neural networks for image reconstruction are also criticized for lacking mathematical justifications and convergence guarantee.

In this work, we leverage the framework developed in \cite{chen2020learnable} and propose an improved learned descent algorithm ELDA. It further boosts image reconstruction quality using an adaptive non-local feature regularizer. More importantly, compared to \cite{chen2020learnable}, ELDA is more computationally efficient since the safeguard iterate is only computed when a descent condition fails to hold, which happens rarely due to allowance of inexact gradient computation in our algorithmic design.
As a result, our model retains convergence guarantee and meanwhile also improves reconstruction quality over existing methods.
%
%
The main contributions of this work are summarized as follows.
\begin{itemize}[leftmargin=*]
    \item We propose an efficient learned descent algorithm {with inexact gradients}, called ELDA, to solve the non-smooth non-convex optimization problem in low-dose CT reconstruction. We provide {comprehensive convergence and iteration complexity analysis} of ELDA.

    \item ELDA adopts efficient update scheme which only computes safeguard iterate when the desired descent condition fails to hold, and hence is more computationally economical than LDA developed in \cite{chen2020learnable}. Moreover, ELDA employs sparsity enhancing and non-local smoothing regularization which further improve imaging quality.

    \item We conduct comprehensive experimental analysis of ELDA and compare it to several state-of-the-art deep-learning based methods for LDCT reconstruction.
\end{itemize}

In section \ref{sec:related_works}, we present the related works in the literature that associate with our problem. Then in section \ref{sec: method}, we present our method by first defining our model and each of its components, and then stating the algorithm and details of network training. After that in section \ref{sec:convergence}, we state our lemma and theorem regarding the output of the network. Section \ref{experiment} presents the numerical results including parameter study, ablation study and comparison with other competing algorithms.


\section{Related Works} \label{sec:related_works}

A natural application of neural networks in CT reconstruction, has been in noise removal in either the projection domain \cite{lee2018deep, proj_domain-lee2017view} or the image domain \cite{CNN4, jin2017deep, CNN3, kang2017deep, xie2020artifact}. In particular, Residual Encoder-Decoder Convolutional Neural Network (RED-CNN) proposed by Chen et al \cite{CNN4}, is an end-to-end mapping from low-dose CT images to normal dose which uses FBP to get low-dose CT images from projections and restrict the problem to denoising in the image domain.
%
And yet another attempt is FBPConvNet by Jin et al \cite{jin2017deep} which is inspired by U-net \cite{ronneberger2015u} and further explores CNN architectures while noting the parallels with the general form of an iterative proximal update.
Model Based Image Reconstruction (MBIR) methods attempt to model CT physics, measurement noise, and image priors in order to achieve higher reconstruction quality in LDCT. Such methods learn the regularizer and are able to improve LDCT reconstruction significantly \cite{iterative-1, zheng2018pwls, chun2017convolutional, ye2019spultra}, however their convergence speed is not optimal \cite{chun2019convolutional}. Later, researchers adopted NNs in other aspects of the algorithm and formed a new class of methods called Iterative Neural Networks (INN). INNs seek to enjoy the best of both world of MBIR and denoising deep neural networks, by employing moderate complexity denoisers for image refining and learning better regularizers \cite{ sun2016deep, chun2019bcd, chun2019momentum,ye2020momentum}.


INNs have network architectures that are \emph{inspired} by the optimization model and algorithm and this learning capacity enables them to outperform the classical iterative solutions by learning better regularizer while also being more time efficient.
For example, recently BCD-Net \cite{chun2019bcd} improved the reconstruction accuracy compared to MBIR methods and NN denoisers. It showed that it generalizes better than denoisers such as FBPConvNet which lack MBIR modules, and also its learned transforms help to outperform state-of-the-art MBIR methods. Further research in this area has been devoted to improving time efficiency of the algorithm with the image quality. Recently,\cite{chun2019momentum} proposed Momentum-Net, as the first fast and convergent INN architecture inspired by Block Proximal Extrapolated Gradient method using a Majorizer. 
It also guarantees convergence to a fixed-point while improving MBIR speed, accuracy, and reconstruction quality. Momentum-Net is a general framework for inverse problems and its recent application to LDCT \cite{ye2020momentum} showed it improves image reconstruction accuracy compared to a state-of-the-art noniterative image denoising NN. Convergence guarantee is one of the main challenges in the design of INNs and beside its theoretical value, it is highly desirable in medical applications.
LEARN\cite{chen2018learn} is another model that unrolls an iterative reconstruction scheme while modeling the regularization by a field-of-experts. 
And yet another similar attempt is Learned Primal-Dual \cite{adler2018learned} which unfolds a proximal  primal-dual optimization  method where the proximal operator is replaced with a CNN.
Their choice of iterative scheme is primal dual hybrid  gradient (PDHG) which is further modified to benefit from the learning capacity of NNs and then used to solve the TV regularized CT reconstruction problem.

In all of the previous works, the architecture is \emph{only inspired} by the optimization model, and in order to improve their performance they introduce components in the network that does not correspond to steps of the algorithm.
Also, the choice of regularization limits the network to only learn local features and as we will empirically demonstrate, it limits the performance of these networks.
One model that attempts to learn non-local features is MAGIC \cite{xia2020magic}. It is also a deep neural network inspired by a simple iterative reconstruction method, i.e. gradient descent. However MAGIC breaks the correspondence between architecture and algorithm in order to extract non-local features. They manually add a non-local corrector in iteration steps which is only \textit{intuitively} justified, and  does not directly correspond to a modified regularizer in the optimization model.

In \cite{chen2020learnable}, a Learned Descent Algorithm (LDA) is developed. The LDA architecture is fully determined by the algorithm and thus the network is fully interpretable.
As interpretability and convergence guarantee is highly desirable in medical imaging, this framework is a promising method for inverse problems such as LDCT reconstruction.
Compared to \cite{chen2020learnable}, the present work proposes a more efficient numerical scheme of LDA, leading to comparable network parameters, lower computational cost, and more stable convergence behavior.
We {achieve} this by developing an efficient learned inexact descent algorithm which only computes the safeguard iterate when a prescribed descent condition fails to hold and thus substantially reduces computational cost in practice. Additionaly, we propose a novel non-local smoothing regularizer that further confirms the heuristics in optimization inspired networks such as MAGIC \cite{xia2020magic} but leads to a fully interpretable network and allows us to provide convergence guarantee of the network.

\section{Method} \label{sec: method}
In this section, we introduce the proposed inexact learned descent algorithm for solving the following  low-dose CT reconstruction model:
%
\begin{equation}\label{eq:loa}
 \mathbf{x}^{(s)}_{\theta} = \argmin_{\xbf}\, \{\phi(\mathbf{x}; \bbf^{(s)}, \theta) := f(\xbf; \bbf^{(s)}) + r(\xbf; \theta) \},
\end{equation}
where $f$ is the data fidelity term that measures the consistency between the reconstructed image $\mathbf{x}$ and the sinogram measurements $\mathbf{b}$, and $r$ is the regularization that may incorporate prior information of $\mathbf{x}$. The regularization $r(\cdot; \theta)$ is realized as a highly structured DNN with parameter $\theta$ to be learned. The optimal parameter $\theta$ of $r$ is then obtained by minimizing the loss function $\mathcal{L}$, where $\mathcal{L}$ measures the (averaged) difference between $\xbf_{\theta}^{(s)}$-the minimizer of $\phi(\cdot; \bbf^{(s)},\theta)$, and the given ground truth $\hat{\xbf}^{(s)}$ for every $s \in [N]$, where $N$ is the number of training data pairs.
For notation simplicity, we write $f(\xbf)$ and $r(\xbf)$ instead of $f(\xbf; \bbf^{(s)})$ and $r(\xbf; \theta)$ respectively hereafter. 
We choose $f(\xbf) = \frac{1}{2} \| A \xbf - \bbf \| ^2$ as the data-fidelity term, where $A$ is the system matrix for CT scanner. However, our proposed method can be readily extended to any smooth but (possibly) nonconvex $f$.

\subsection{Regularization Term in  Model \eqref{eq:loa}}
The regularization term $r$ in \eqref{eq:loa} consists of two parts. One of them enhances the sparsity of the solution under a learned transform and the other one smooths the feature maps non-locally:
\begin{equation}\label{eq:r}
  r(\xbf) := \hat{r}(\xbf) + \lambda  \overline{r}(\xbf),
\end{equation}
where $\lambda$ is a coefficient to balance these two terms which can be learned.

\subsubsection{The sparsity-enhancing regularizer}
\label{subsubsect:local_sparsity}

To enhance the sparsity of $\xbf$
under a learned transform $\gbf$,
we propose to minimize  the $l_{2,1}$ norm of $\gbf(\xbf)$. If $\gbf$ is a differential operator, then the $l_{2,1}$ norm of $\gbf(\xbf)$ reduces to the total variation of $\xbf$.
%
That is,
\begin{equation}\label{eq:r_1}
  \hat{r}(\xbf) = \|\gbf(\xbf)\|_{2,1} = \sum_{i = 1}^{m} \|\gbf_i(\xbf)\|,
\end{equation}
 where each $\gbf_i(\xbf) \in \mathbb{R}^d$ can be viewed as a feature descriptor vector at the position $i$, as depicted in Fig. \ref{fig:fold_feature_map} (up).
In our experiments, we simply set the feature extraction operator $\gbf$ to a vanilla $l$-layer CNN with nonlinear activation function $\sigma$ but no bias, as follows:
\begin{equation}\label{eq:g}
  \gbf(\xbf) = \wbf_l * \sigma \cdots \ \sigma ( \wbf_3 * \sigma ( \wbf _2 * \sigma ( \wbf _1 * \xbf ))),
\end{equation}
where $\{\wbf _q \}_{q = 1}^{l}$ denote the convolution weights consisting of $d$ kernels with identical spatial kernel size ($3 \times 3$), and $*$ denotes the convolution operation.
Here, the componentwise activation function $\sigma$ is constructed to be the smoothed rectified linear unit as defined below
\begin{equation}\label{eq:sigma}
\sigma (x) =
\begin{cases}
0, & \mbox{if} \ x \leq -\delta, \\
\frac{1}{4\delta} x^2 + \frac{1}{2} x + \frac{\delta}{4}, & \mbox{if} \ -\delta < x < \delta, \\
x, & \mbox{if} \ x \geq \delta,
\end{cases}
\end{equation}
where the prefixed parameter $\delta$ is set to be $0.001$ in our experiment.
Besides the smooth $\sigma$, each convolution operation of $\gbf$ in \eqref{eq:g} can be viewed as matrix multiplication, which enable $\gbf$ to be differentiable, and $\nabla \gbf$ can be easily obtained by Chain Rule where each $\wbf_q^{\top}$ can be implemented as transposed convolutional operation \cite{dumoulin2016guide}.

As $\hat{r}(\xbf)$ defined in \eqref{eq:r_1} is nonsmooth and nonconvex, we apply the Nesterov's smoothing technique \cite{nesterov2005smooth} to
get the smooth approximation and the detail is given in \cite{chen2020learnable}:
\begin{equation}\label{eq:r_eps_hat}
 \hatreps(\xbf)  = \sum_{i \in I_0} \frac{1}{2\varepsilon}  \|\gbf_i(\xbf)\|^2 + \sum_{i \in I_1} \del[2]{\|\gbf_i(\xbf)\| - \frac{\varepsilon}{2} },
\end{equation}
where
$I_0 = \{ i \in [m] \ \vert \ \|\gbf_i(\xbf)\| \leq \varepsilon \}, \  I_1 = [m] \setminus I_0.
$
Here the parameter $\varepsilon$ controls how close the smoothed $ \hatreps(\xbf) $ is to the original function $\hat{r}(\xbf)$, and
one can readily show that
$
\hatreps(\xbf)\leq \hat{r}(\xbf) \leq \hatreps(\xbf) +\frac{m\varepsilon}{2} $ for all $\xbf$ in $\mathbb{R}^n$.
From \eqref{eq:r_eps_hat} we can also derive $\nabla \hatreps(\xbf)$ to be
$$\label{eq:d_r_eps}
\nabla \hatreps(\xbf) = \sum_{i \in I_0} \nabla \gbf_i(\xbf)^{\top} \frac{\gbf_i(\xbf)}{\varepsilon}   + \sum_{i \in I_1} \nabla \gbf_i(\xbf)^{\top}  \frac{\gbf_i(\xbf)}{\|\gbf_i(\xbf)\|} ,
$$
where $\nabla \gbf_i(\xbf)\in \mathbb{R}^{d \times n}$ is the Jacobian of $\gbf_i$ at $\xbf$.


%
\begin{figure}[t]
\centering
\includegraphics[width=0.5\textwidth]{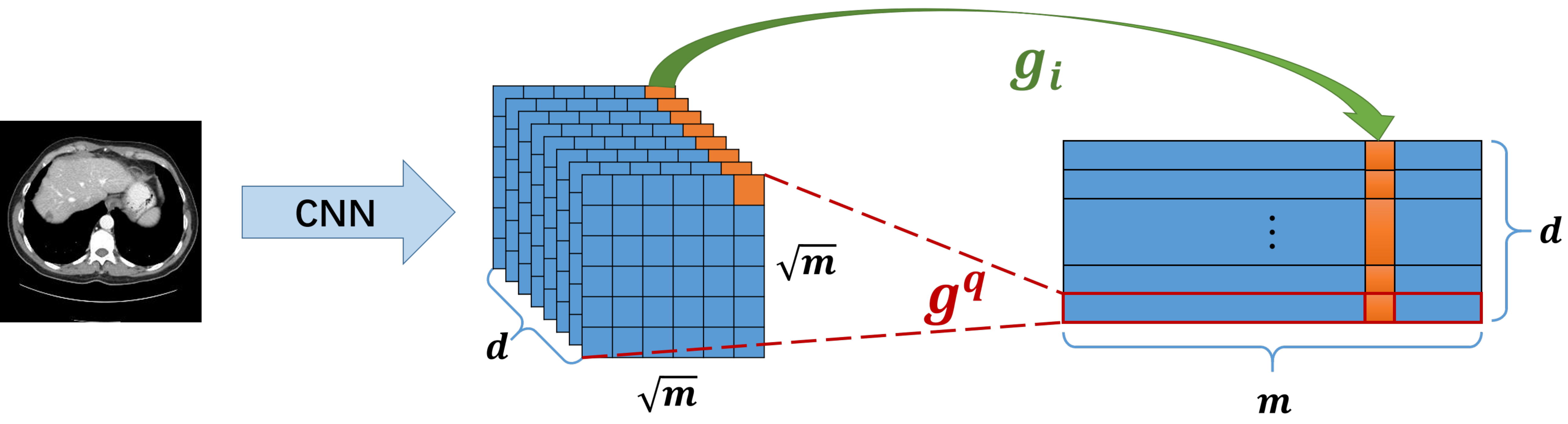}
\includegraphics[width=0.5\textwidth]{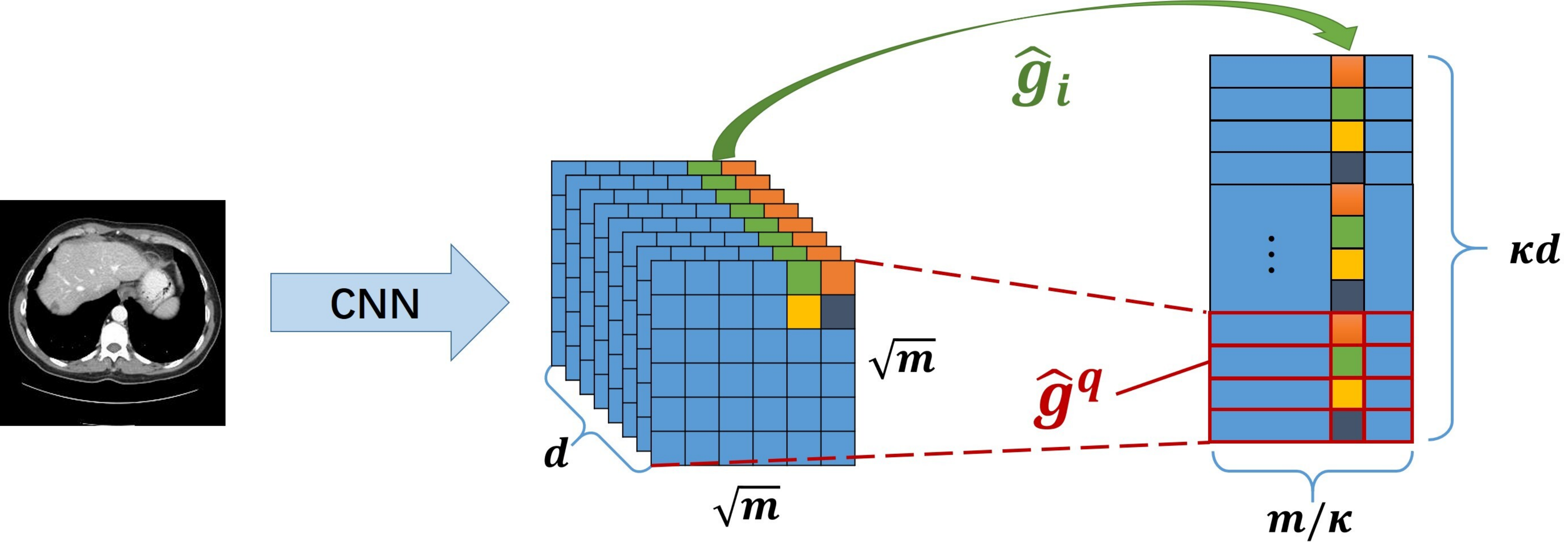}
\caption{The feature matrix $\gbf \in \mathbb{R}^{d \times m}$ (up) and the folded feature matrix $\hat{\gbf} \in \mathbb{R}^{\kappa d \times \frac{m}{\kappa}}$ (bottom) reshaped from the feature maps obtained from the last convolution of the CNN defined in \eqref{eq:g}. The folding rate $\kappa$ is 4 for $\hat{\gbf}$ in this illustration.}
\label{fig:fold_feature_map}
\end{figure}
\subsubsection{The nonlocal smoothing regularizer}
\label{subsubsect:nonlocal_smoothing}

Since convolution operations only extract the local information, each feature descriptor vector $\gbf_i$ can only encode the local features of a small patch of the input $\xbf$ (i.e. receptive field) \cite{LeB17}.
So here we seek to incorporate an additional non-local smoothing regularizer $\overline{r}$  that enables capturing of the underlying long-range dependencies between the \emph{patches} of the feature descriptor vectors.
%
To this end
we form the folded feature descriptor vectors $\{\hat{\gbf}_i\}$ as described in  Fig. \ref{fig:fold_feature_map} (bottom) by folding the adjacent $\kappa$ feature descriptors together, and
define $\overline{r}$ by:

\begin{equation}\label{eq:graph_qua}
  \overline{r}(\xbf)= \sum_{(i,j)} \mathcal{W}_{ij} \|\hat{\gbf}_i(\xbf) - \hat{\gbf}_j(\xbf)\|^2,
\end{equation}
where the similarity matrix $\mathcal{W}$ is defined by $\mathcal{W}_{ij} = \exp{\Big(-\frac{\|\hat{\gbf}_i(\xbf) - \hat{\gbf}_j(\xbf)\|^2}{\delta^2}\Big)},$
and $\delta$ is the standard deviation, which is estimated by the median of the Euclidean distances between the folded feature descriptor vectors in the model.

Additionally, $\overline{r}$ can also be written in the quadratic form \cite{6789755} as
$
\overline{r}(\xbf) =  tr(\hat\gbf(\xbf)\mathcal{L} \ \hat\gbf(\xbf)^{\top} ),
$
where $tr()$ is the trace operator, $\mathcal{L}= \mathcal{D} - \mathcal{W}$, and $\mathcal{D}$ is the diagonal matrix with $\mathcal{D}_{ii} = \sum_{j = 1}^{m} \mathcal{W}_{ij}$, and $\mathcal{L}$ is positive semidefinite. And its gradient is computed by
\begin{subequations}
\begin{align*}
\nabla \overline{r}(\xbf) &= 2 \cdot \sum_{(i,j)} \tilde{\mathcal{W}}_{ij} (\nabla \hat\gbf_i(\xbf) - \nabla \hat\gbf_j(\xbf))^\top (\hat\gbf_i(\xbf) - \hat\gbf_j(\xbf)) \\
&= 2 \cdot \sum_{q = 1}^{\kappa d} \nabla\hat\gbf^q(\xbf)^{\top}\tilde{\mathcal{L}} \ \hat\gbf^q(\xbf),  \label{eq:graph_qua_diff_true}
\end{align*}
\end{subequations}
where $\nabla \hat\gbf^q(\xbf)\in \mathbb{R}^{ \frac{m}{\kappa} \times n}$ is the Jacobian of $\hat\gbf^q(\xbf)$. And $\tilde{\mathcal{L}} = \tilde{\mathcal{D}} - \tilde{\mathcal{W}}$, $\tilde{\mathcal{D}}_{ii} = \sum_{j = 1}^{m} \mathcal{\tilde W}_{ij}$ and $\tilde{\mathcal{W}}_{ij} = \mathcal{W}_{ij}(\xbf) (1 - \frac{ \lVert \hat\gbf_i(\xbf) - \hat\gbf_j(\xbf) \rVert^2  }{\delta^2})$.
Each $\hat\gbf^q(\xbf)$ represents the $q$-th row of the folded feature matrix $\hat\gbf(\xbf)$ as illustrated in Fig. \ref{fig:fold_feature_map}.


%
\subsection{Inexact Learned Descent Algorithm}
\label{subsec:LDA}
Now we present an inexact smoothing gradient descent type algorithm to solve the nonconvex and nonsmooth problem \eqref{eq:loa} with the smooth approximation of $r(\xbf)$ defined by
$
\reps(\xbf) :=  \hatreps(\xbf) + \lambda  \overline{r}(\xbf).
$
The proposed algorithm is shown in Algorithm \ref{alg:lda}. In each iteration $k$, we solve the following smoothed problem \eqref{phi_smooth} with fixed $\varepsilon = \varepsilon_k$ in Line 3-14. And Line 15 is aimed to check and update $\varepsilon_k$ by a reduction  principle.
\begin{equation} \label{phi_smooth}
 \min_{\xbf}\, \{\phi_{\varepsilon}(\mathbf{x}; \bbf^{(s)}, \theta) := f(\xbf; \bbf^{(s)}) + \reps(\xbf; \theta) \}.
\end{equation}

As the regularization term $\reps$ is learned via a deep neural network (DNN), some common issues of the DNN have to be taken into consideration when designing the algorithm, such as gradient exploding and vanishing problem during training \cite{DBLP:journals/corr/HeZR016}.
%
%
Substantial improvement in performance has been achieved by ResNet \cite{ResNet} which introduces residual connections to alleviate these issues.
As in \eqref{phi_smooth} only the second term $\reps(\xbf; \theta)$ is learned, we desire to have individual residual updates for this term in our algorithm.
To this end, we use the first order proximal method to solve the smoothed problem \eqref{phi_smooth} by iterating the following steps
\begin{subequations}
\begin{align}
\zkp &= \xk - \alpha_k \nabla f(\xk),\\
\xkp &= \prox_{\alpha_k r_{\varepsilon_{k}}}(\zkp),  \label{eq:prox}
\end{align}
\end{subequations}
where $\prox_{\alpha r}(\zbf) := \argmin_{\xbf}\ \frac{1}{2\alpha} \| \xbf - \zbf \| ^2 + r(\xbf)$.

From our construction of $r_{\varepsilon_{k}}$, it is hard to get the close-form solution to subproblem in \eqref{eq:prox}. Here we propose to linearize the \enquote{nonsimple} term $r_{\varepsilon_{k}}$ by
\begin{equation}
    \label{eq:r_approx}
\tilde{r}_{\varepsilon_{k}}(\xbf) = r_{\varepsilon_{k}}(\zkp) + \langle \nabla r_{\varepsilon_{k}} (\zbf_{k+1}) , \xbf - \zbf_{k+1}\rangle +  \frac{1}{2\beta_k} \| \xbf - \zbf _{k+1} \| ^2.
\end{equation}
With this approximation, instead of solving \eqref{eq:prox} directly, we update by the following step
\begin{equation}
    \label{eq:uuu}
    \ukp = \prox_{\alpha_k \tilde{r}_{\varepsilon_{k}}}(\zkp),
\end{equation}
which has a closed-form solution giving the residual update
\begin{equation} \label{eq:u_closed}
    \ubf_{k+1} = \zbf_{k+1} - \tau_k \nabla r_{\varepsilon_{k}} (\zbf_{k+1}),
\end{equation}
where $\nabla r_{\varepsilon_{k}}  = \nabla \hat r_{\varepsilon_{k}} + \lambda \nabla \overline{r}$ and $\tau_k = \frac{\alpha_k \beta_k}{\alpha_k + \beta_k}$.

%
From the optimization perspective, with the approximation in \eqref{eq:r_approx}, if we update $\xkp = \ukp$ then the algorithm can not be guaranteed to converge.
In order to ensure convergence, we check whether $\ukp$ satisfies
\begin{equation} \label{condition:u}
\begin{split}
\| \nabla \phi_{\varepsilon_{k}} (\xbf_k) \| \leq c \| \ukp - \xk \| \ \ \   \mbox{and} \ \ \  \phi_{\varepsilon}(\ukp) - \phi_{\varepsilon}(\xk) \leq - \frac{\iota}{2}\| \ukp - \xk \|^2,
\end{split}
\end{equation}
where $c$ and $\iota$ are prefixed constant numbers.
If the condition \eqref{condition:u} holds, we take $\xkp = \ukp$; otherwise, we take the standard gradient descent $\vbf_{k+1}$ coming from
\begin{multline} \label{eq:v}
\vbf_{k+1} = \argmin_{\xbf} \langle \nabla f(\xbf_{k}) , \xbf - \xbf_{k}\rangle
 + \langle \nabla r_{\varepsilon_{k}} (\xbf_{k}) , \xbf - \xbf_{k}\rangle  +  \frac{1}{2\alpha_k} \| \xbf - \xbf _{k} \| ^2,
\end{multline}
which has the exact solution \begin{equation} \begin{split}
\label{eq:v_closed} \vbf_{k+1} = \xk - \alpha_k \nabla f(\xbf_k) - \alpha_k \nabla r_{\varepsilon_{k}}(\xbf_k).\end{split} \end{equation}

To ensure convergence, we need to find $\alpha_k$ through line search such that $\vkp$ satisfies
%
\begin{equation} \label{condition:v}
\phi_{\epsk}(\vkp) - \phi_{\epsk}(\xbf_k) \le - \tau \| \vkp - \xbf_k\|^2,
\end{equation}
where $\tau$ is a prefixed constant. Lemma \ref{lem:inner} proves the convergence of lines 3--14 Algorithm \ref{alg:lda}, including the termination of its line search (lines 9--13) in finitely many steps.

Our algorithm is inspired by \cite{chen2020learnable} but modified to become more efficient and suitable for deep neural network.
The first contrast is the number of computations of the two candidates $\ukp$ and $\vkp$.
While \cite{chen2020learnable} computes both candidates at every iteration and then chooses the one that achieves a smaller function value, we propose the criteria above in \eqref{condition:u} for updating $\xkp$,
which potentially
saves extra computational time for calculating the candidate $\vkp$.
In addition, the descending condition in Line 5 of Algorithm \ref{alg:lda} mitigates the frequent alternations between the candidates $\ukp$ and $\vkp$ with the algorithm proceeding, as details shown in Section \ref{sec: Ablation Study}.

\subsubsection*{The learned inexact gradient} To further increase the capacity of the network, we employ the learned transposed convolution operator, i.e. we replace $\wbf_q^{\top}$ by a transposed convolution $\widetilde{\wbf}_{q}$ with relearned weights, where $q$ denotes the index of convolution in \eqref{eq:g}.
To approximately achieve $\widetilde{\wbf}_{q} \approx \wbf_q^{\top}$,
we add the constraint term $ \mathcal{L}_{constraint} = \frac{1}{N_w} \sum_{q = 1}^4   \| \widetilde{\wbf}_{q} - \wbf^{\top}_q \| ^2_F$ to the loss function in training to produce the data-driven transposed convolutions.
Here $N_w$ is the number of parameters in learned transposed convolutions and $\|\cdot\|_F$ is the Frobenius norm.
In effect, the consequence of this modification is only to substitute $\nabla r_{\varepsilon_{k}}$ by the inexact gradient $\widetilde{\nabla r}_{\varepsilon_{k}}$ equipped with learned transpose at Line 4 in Algorithm \ref{alg:lda}.
This can further increase the capacity of the unrolled network while maintaining the convergence property.
%
\begin{algorithm}[tbh]
\caption{The Efficient learned Descent Algorithm (ELDA) for the Nonsmooth Nonconvex Problem}
\label{alg:lda}
\begin{algorithmic}[1]
\STATE \textbf{Input:} Initial $\xbf_0$, $0<\rho, \gamma<1$, and $\varepsilon_0,\sigma, c, \iota, \tau>0$. Maximum iteration $K$ or tolerance $\etol>0$.
\FOR{$k=0,1,2,\dots,K$}
\STATE $\zkp = \xk - \alpha_k \nabla f(\xk)$,
\STATE $\ukp = \zkp - \tau_k \nabla r_{\varepsilon_{k}} (\zkp)$, {\small(possibly inexact)}
\IF{ condition \eqref{condition:u} holds} 
\STATE set $\xkp = \ukp$,
\ELSE
\STATE $\vbf_{k+1} = \xk - \alpha_k \nabla f(\xbf_k) - \alpha_k \nabla r_{\varepsilon_{k}}(\xbf_k)$, \label{marker}
\IF{ condition \eqref{condition:v} holds} 
\STATE set $\xkp = \vkp$,
\ELSE
\STATE update $\alpha_k \leftarrow \rho \alpha_k$,
then \textbf{go to}~\ref{marker},
\ENDIF
\ENDIF
\STATE \textbf{if} $\|\nabla \phi_{\varepsilon_k}(\xkp)\| < \sigma \gamma {\varepsilon_k}$,  set $\varepsilon_{k+1}= \gamma {\varepsilon_k}$;  \textbf{otherwise}, set $\varepsilon_{k+1}={\varepsilon_k}$.
\STATE \textbf{if} $\sigma {\varepsilon_k} < \etol$, terminate.
\ENDFOR
\STATE \textbf{Output:} $\xbf_{k+1}$.
\end{algorithmic}
\end{algorithm}
\subsection{Network Training:}
We allow the step sizes $\alpha_k$ and $\tau_k$ to vary in different phases. Moreover, all $\{\alpha_k, \tau_k\}_{k = 1} ^ K$ and initial threshold $\varepsilon_0$ are designed to be learned parameters fitted by data.
Here let $\theta$
stand for the set of all learned parameters of ELDA which consists of the weights of the convolutions and approximated transposed convolutions, step sizes $\{\alpha_k, \tau_k\}_{k = 1} ^ K$ and threshold $\varepsilon_0$, parameter $\lambda$.
Given $N$ training data pairs $\{(\bbf^{(s)}, \hat{\xbf}^{(s)})\}_{s=1}^{N}$ of the ground truth data $\hat{\xbf}^{(s)}$ and its corresponding measurement $\bbf^{(s)}$, the loss function $\mathcal{L}(\theta)$ is defined to be the sum of the discrepancy loss $\mathcal{L}_{discrepancy}$ and the constraint loss $\mathcal{L}_{constraint}$:
\begin{equation}\label{loss}
\begin{split}
  \mathcal{L}(\theta) = \underbrace{\frac{1}{N}\sum_{s = 1}^N \|\xbf^K_\theta - \hat{\xbf}^{(s)}\|^2}_{\mathcal{L}_{discrepancy}}
  + \underbrace{\frac{\vartheta}{N_w} \sum_{q = 1}^4  \| \widetilde{\wbf}_{q} - \wbf^{\top}_q \| ^2_F }_{\mathcal{L}_{constraint}},
  \end{split}
\end{equation}
where $\mathcal{L}_{discrepancy}$ measures the discrepancy between the ground truth $\hat{\xbf}^{(s)}$ and $\xbf^K_\theta$ which is the output of the $K$-phase network.
Here, the constraint coefficient $\vartheta$ is set to $10^{-2}$ in our experiment.
\section{Convergence Analysis}
\label{sec:convergence}
According to the problem we are solving, we make a few assumptions on $f$ and $\gbf$ throughout this section.
\\
(A1): $f$ is differentiable and (possibly) nonconvex, and $\nabla f$ is $L_f$-Lipschitz continuous.
\\
(A2): Every component of $\gbf$ is differentiable and (possibly) nonconvex, $\nabla \gbf$ is $L_g$-Lipschitz continuous, and $\sup_{\xbf \in \Xcal} \| \nabla \gbf(\xbf)\| \leq M $ for some constant $M>0$.
\\
(A3): $\phi$ is coercive, and $\phi^* = \min_{\xbf \in \Xcal} \phi(\xbf) > -\infty$.

With the smoothly differentiable activation $\sigma$ defined in \eqref{eq:sigma} and boundedness of $\sigma'$ as well as the fixed convolution weights in \eqref{eq:g} after training, we can immediately verify that the first two assumptions hold,
and typically in image reconstruction $\phi$ is assumed to be coercive  \cite{chen2020learnable}.

As the objective function in \eqref{eq:loa} is nonsmooth and nonconvex, we utilize the Clarke subdifferential\cite{Clarke83} to characterize the optimality of solutions.
We denote $D(\zbf;\vbf) := \limsup_{\zbf \rightarrow \xbf,\, t \downarrow 0} [(f(\zbf+t \vbf)-f(\zbf))/t]$.
\begin{definition}[Clarke subdifferential]
Suppose that $f: \mathbb{R}^{n} \to (-\infty,+\infty]$  is locally Lipschitz, the Clarke subdifferential $ \partial f(\xbf)$ of $f$  at $\xbf$  is defined as
\vspace{-5pt}
\begin{equation*}
\partial f(\xbf) := \Big\{\wbf \in \mathbb{R}^{n}\ \vert \ \langle \wbf, \vbf \rangle \leq D(\zbf;\vbf), \forall\, \vbf \in \mathbb{R}^{n} \Big\}.
\end{equation*}
\end{definition}
\vspace{3.5pt}
\begin{definition}[Clarke stationary point]
\label{def:clarke_cp}
For a locally Lipschitz function $f$, a point $\xbf \in R^n$ is called a Clarke stationary point of $f$ if $0 \in \partial f(\xbf)$.
\end{definition}
\begin{lemma}
The gradient of $\overline{r}$ is Lipschitz continuous.
\end{lemma}

\begin{proof}

From equation 6 in the paper, it follows that
\begin{equation*}
    \nabla_{(\xbf)} \overline{r} = \sum_{(i,j)} 2 \exp{\Big(-\frac{\|\gbf_i(\xbf) - \gbf_j(\xbf)\|^2}{\sigma^2}\Big)} (1 - \frac{ \lVert \gbf_i(\xbf) - \gbf_j(\xbf) \rVert^2  }{\sigma^2}) (\nabla \gbf_i(\xbf) - \nabla \gbf_j(\xbf))^\top (\gbf_i(\xbf) - \gbf_j(\xbf)).
\end{equation*}

We only need to show one of the terms under the sum is Lipschitz, i.e. we need $f(\ubf) = \exp{\Big(-\frac{\| \ubf \|^2}{\sigma^2}\Big)} (1 - \frac{ \lVert \ubf \rVert^2  }{\sigma^2}) (\nabla \ubf )^\top (\ubf)$ to be Lipschitz,
where $\ubf = \gbf_i(\xbf) - \gbf_j(\xbf) \in \mathbb{R}^d$. We show $f(\ubf)$ is Lipschitz by showing its derivative is bounded:
\begin{align*}
    \lVert \nabla f(\ubf) \rVert
    &\leqslant
    \lVert \exp{\Big(-\frac{\| \ubf \|^2}{\sigma^2}\Big)} \big(- \frac{2}{\sigma^2} (\nabla \ubf )^\top (\ubf) \big) \rVert \,
     (1 - \frac{ \lVert \ubf \rVert^2  }{\sigma^2}) \,
    \lVert (\nabla \ubf )^\top \rVert
    \lVert \ubf \rVert
    +
    \exp{\Big(-\frac{\| \ubf \|^2}{\sigma^2}\Big)}  \,
    \lVert \big(- \frac{2}{\sigma^2} (\nabla \ubf )^\top (\ubf) \big) \rVert \,
    \lVert (\nabla \ubf )^\top \rVert
    \lVert \ubf \rVert \\
    &+
    \lvert \exp{\Big(-\frac{\| \ubf \|^2}{\sigma^2}\Big)} \,
    (1 - \frac{ \lVert \ubf \rVert^2  }{\sigma^2}) \rvert  \,
    \lVert (\nabla^2 \ubf ) \rVert
    \lVert \ubf \rVert
    +
    \lvert \exp{\Big(-\frac{\| \ubf \|^2}{\sigma^2}\Big)}
    (1 - \frac{ \lVert \ubf \rVert^2  }{\sigma^2}) \rvert
    \lVert \nabla \ubf \rVert^2.
\end{align*}
Note that $\ubf = \gbf_i(\xbf) - \gbf_j(\xbf) $ is Lipschitz in $\xbf$ since $\sup_{\xbf \in \Xcal} \| \nabla \gbf(\xbf)\| \leq M $, therefore $\lVert \nabla \ubf \rVert \leqslant M$. Also since $\nabla \gbf$ is $L_g$ Lipschitz, we get $\lVert \nabla^2 \ubf \rVert \leqslant 2 L_g$.
Also, a polynomial in $\| \ubf \|$ times $\exp{\Big(-\frac{\| \ubf \|^2}{\sigma^2}\Big)}$ is bounded so we get:
\begin{equation}
    \lVert \nabla f(\ubf) \rVert \leqslant (4\frac{M^2}{\sigma^2} + L_g + M^2) C.
\end{equation}
where $C$ is some constant that bounds all the instances of polynomial in $\| \ubf \|$ times $\exp{\Big(-\frac{\| \ubf \|^2}{\sigma^2}\Big)}$ occurring above. Therefore $\nabla_{(\xbf)} \overline{r}$ is Lipschitz with constant $L_r = 2m^2(4\frac{M^2}{\sigma^2} + L_g + M^2) C$.
\end{proof}

The following Lemma \ref{lem:inner} considers the case where $\varepsilon$ is a positive constant, which corresponds to an iterative scheme that only executes Lines 3--14 of Algorithm \ref{alg:lda}.

\begin{lemma}\label{lem:inner}
Let $\varepsilon, c, \iota, \tau>0$, $0< \rho <1$ and arbitrary initial $\xbf_0 \in \mathbb{R}^n$. Suppose $\{\xbf_{k}\}$ is the sequence generated by repeating Lines 3--14 of Algorithm \ref{alg:lda} with fixed $\varepsilon_k = \varepsilon$, and $\phi^*:=\min_{\xbf \in \mathbb{R}^n} \phi(\xbf)$. Then $\| \nabla \phi_{\varepsilon}(\xbf_k) \| \to 0$ as $k\to \infty$.
\end{lemma}

\begin{proof}
In each iteration we compute $\ukp$ by $\zbf_{k+1} - \tau_k \nabla r_{\varepsilon_{k}} (\zbf_{k+1})$ 
, and put $\xkp = \ukp$ only if the condition $\{ \| \phi_{\varepsilon_{k}} (\xbf_k) \| \leq c \| \ukp - \xk \|$ $and$ $ \phi_{\varepsilon}(\ukp) - \phi_{\varepsilon}(\xk) \leq - \frac{\iota}{2}\| \ukp - \xk \|^2\}$ is satisfied. In this case, we will have $ \phi_{\varepsilon}(\xkp)  \leq \phi_{\varepsilon}(\xk) $.

Otherwise we compute $\vkp = \xk - \alpha_k \nabla f(\xbf_k) - \alpha_k \nabla r_{\varepsilon_{k}}(\xbf_k) = - \alpha_k \nabla \phi_{\varepsilon}(\xkp)$ where $\alpha_k$ is found through line search until the criteria

\begin{equation} \label{eq:criteria8}
    \phi_{\epsk}(\vkp) - \phi_{\epsk}(\xbf_k) \le - \tau \| \vkp - \xbf_k\|^2
\end{equation}
is satisfied and then put $\xkp = \vkp$. In this case,
from Lemma 3.2 in \cite{chen2020learnable} we know that the gradient $\nabla \reps$  is Lipschitz continuous with constant $\sqrt{m} L_g+\frac{M^2}{\varepsilon}$. And $\nabla \overline{r}$ is also Lipschitz continuous with constant $L_r$.
Furthermore, in (A1) we assumed $\nabla f$ is $L_f$-Lipschitz continuous. Hence putting $L_{\varepsilon} = L_f + \sqrt{m} L_g + \frac{M^2}{\varepsilon} + L_r$, we get that $\nabla \phi_{\varepsilon}$ is $L_\varepsilon$-Lipschitz continuous, which implies
\begin{equation}\label{lipF}
  \phi_{\varepsilon}(\vbf_{k+1}) \leq \phi_{\varepsilon}(\xbf_{k}) + \langle\nabla \phi_{\varepsilon}(\xbf_{k}), \vbf_{k+1}-\xbf_{k}\rangle+\frac{L_{\varepsilon}}{2} \|  \vbf_{k+1}-\xbf_{k} \| ^2.
\end{equation}

Also by the optimality condition of $\vbf_{k+1} = \argmin_{\xbf}\ \langle \nabla f(\xbf_{k}) , \xbf - \xbf_{k}\rangle  \\ + \langle \nabla r_{\varepsilon_{k}} (\xbf_{k}) , \xbf - \xbf_{k}\rangle  +  \frac{1}{2\alpha_k} \| \xbf - \xbf _{k} \| ^2,$ we have:
\begin{equation}\label{eq:v_opt}
    \langle\nabla \phi_{\varepsilon}(\xbf_{k}), \vbf_{k+1}-\xbf_{k}\rangle+\frac{1}{2\alpha_k} \|  \vbf_{k+1}-\xbf_{k} \| ^2 = \frac{-1}{2 \alpha_k} \|  \vbf_{k+1}-\xbf_{k} \| ^2 \leq 0.
\end{equation}
From above \eqref{lipF} and \eqref{eq:v_opt}, we can get
\begin{equation}\label{eq:diff}
\phi_{\varepsilon}(\vbf_{k+1})-\phi_{\varepsilon}(\xbf_{k})\leq - \del[2]{\frac{1}{\alpha_k} -\frac{L_{\varepsilon}}{2}} \|  \vbf_{k+1}-\xbf_{k} \|^2.
\end{equation}
Therefore it is enough for $\alpha_k \leqslant \frac{1}{\tau + L_\epsilon /2}$ so that the criteria \eqref{eq:criteria8} is satisfied. This implies that there is a finite $k$ so that $\rho^k \alpha_k \leqslant \frac{1}{\tau + L_\epsilon /2}$ and the line search stops successfully and we can get $\phi_{\varepsilon}(\xbf_{k+1})=\phi_{\varepsilon}(\vbf_{k+1}) \leq \phi_{\varepsilon}(\xbf_{k})$.
Hence in both cases (i) and (ii), we can conclude that $\phi_{\varepsilon}(\xbf_{k+1}) \leq \phi_{\varepsilon}(\xbf_{k})$.

Now from $\vbf_{k+1} = \xk - \alpha_k \nabla \phi_{\varepsilon}(\xbf_k)$ we have
\begin{equation}
    \phi_{\varepsilon}(\vbf_{k+1})-\phi_{\varepsilon}(\xbf_{k}) \leqslant
    - \frac{\tau}{\alpha_k ^ 2} \lVert \nabla \phi_{\varepsilon}(\xbf_k) \rVert^2.
\end{equation}
which by rearranging gives
\begin{equation}\label{eq:diff2}
 \|  \nabla \phi_{\varepsilon} (\xbf_k)\|^2 \leq \frac{\alpha_k^2}{\tau} (\phi_{\varepsilon}(\xbf_{k}) - \phi_{\varepsilon}(\vbf_{k+1})).
\end{equation}
And From $ \| \nabla \phi_{\varepsilon_{k}} (\xbf_k) \| \leq c \| \ukp - \xk \|$ and $ \phi_{\varepsilon}(\ukp) - \phi_{\varepsilon}(\xk) \leq - \frac{\iota }{2}\| \ukp - \xk \|^2 $, we can get
\begin{equation}\label{eq:diffu}
 \|  \nabla \phi_{\varepsilon} (\xbf_k)\|^2 \leq \frac{2 c^2}{\iota} (\phi_{\varepsilon}(\xbf_{k}) - \phi_{\varepsilon}(\ubf_{k+1})).
\end{equation}
Because $\xbf_{k+1}$ equals either $\ubf_{k+1}$ or $\vbf_{k+1}$ depending on the condition on Line 5 and 9 in the Algorithm, from \eqref{eq:diff2} and \eqref{eq:diffu}, we have
 \begin{equation}\label{eq:diff3}
 \|  \nabla \phi_{\varepsilon} (\xbf_k)\|^2 \leq C (\phi_{\varepsilon}(\xbf_{k}) - \phi_{\varepsilon}(\xbf_{k+1})).
\end{equation}
where $C = \max \{ \frac{\alpha_k^2}{\tau}, \frac{2 c^2}{\iota} \}$.

Summing up \eqref{eq:diff3} for $k=0,\dots,K$, we have
\begin{equation}
    \sum_{k=0}^{K} \|  \nabla \phi_{\varepsilon} (\xbf_k) \| ^2 \leq C(\phi_{\varepsilon}(\xbf_0) - \phi_{\varepsilon}(\xbf_{K+1})).
\end{equation}
Combining with the fact $\phi_{\varepsilon}(\xbf) \ge \phi(\xbf)-\frac{m\varepsilon}{2} \ge \phi^* - \frac{m\varepsilon}{2}$, we have
\begin{equation}
    \sum_{k=0}^{K} \|  \nabla \phi_{\varepsilon} (\xbf_k) \| ^2
    \le C(\phi_{\varepsilon}(\xbf_0) - \phi^* + \frac{m\varepsilon}{2}).
\end{equation}
The right hand side is finite, and thus by letting $K\to\infty$ we conclude $\| \nabla \phi_{\varepsilon} (\xbf_k) \| \to 0$, which proves the lemma.
\end{proof}

Note that Lemma \ref{lem:inner} does not rely on the term $\nabla r_{\varepsilon_{k}} (\zbf_{k+1})$ in \eqref{eq:u_closed} and Line 5 in Algorithm \ref{alg:lda}. Thus it will still hold if we change the $\nabla r_{\varepsilon_{k}} (\zbf_{k+1})$ to a generalized learned inexact $\widetilde{\nabla r}_{\varepsilon_{k}}(\zbf_{k+1})$ when updating $\ukp$ as described in Section \ref{subsec:LDA}.
%
Next we consider the case where $\varepsilon$ varies in Theorem \ref{theorem_convergence}. More precisely, we focus on the subsequence $\{\xbf_{\kl+1}\}$ which selects the iterates when the reduction criterion in Line 15 is satisfied for $k = \kl$ and $\varepsilon_k$ is reduced.
%

\begin{lemma}
\label{lem:phi_decay}
Suppose that the sequence $\{\xbf_k\}$ is generated by Algorithm 1 and any initial $\xbf_0$. Then for any $k\ge 0$ we have
\begin{equation}\label{eq:phi_decay}
    \phi_{\varepsilon_{k+1}}(\xbf_{k+1}) + \frac{m \varepsilon_{k+1}}{2} \le \phi_{\varepsilon_{k}}(\xbf_{k+1}) + \frac{m \varepsilon_{k}}{2} \le \phi_{\varepsilon_{k}}(\xbf_{k}) + \frac{m \varepsilon_{k}}{2}.
\end{equation}
\end{lemma}

\begin{proof} The proof of this lemma is similar to Lemma 3.4 of \cite{chen2020learnable}. The second inequality is immediate from \eqref{eq:diff3}. So we focus on the first inequality.
For any $\varepsilon>0$ and $\xbf$, denote
\begin{equation}\label{eq:rei}
r_{\varepsilon, i}(\xbf) := \begin{cases}
\frac{1}{2 \varepsilon} \| \gbf_i(\xbf)\|^2, & \mbox{if} \ \|\gbf_i(\xbf)\| \le \varepsilon, \\
\| \gbf_i(\xbf)\| - \frac{\varepsilon}{2}, & \mbox{if} \ \|\gbf_i(\xbf)\| > \varepsilon .
\end{cases}
\end{equation}
Since $\phi_{\varepsilon}(\xbf) = \sum_{i=1}^m r_{\varepsilon,i}(\xbf) + \lambda \overline{r}(\xbf) + f(\xbf)$, to prove the first inequality it suffices to show that
\begin{equation}
    \label{eq:r_decay}
    r_{\varepsilon_{k+1},i}(\xkp) + \frac{\varepsilon_{k+1}}{2} \le r_{\epsk,i}(\xkp) + \frac{\epsk}{2}.
\end{equation}
If $\varepsilon_{k+1} = \varepsilon_{k}$, then the two quantities above are identical and the first inequality holds. Now suppose $\varepsilon_{k+1} = \gamma \varepsilon_{k} < \varepsilon_k$. %
We then consider the relation between $\| \gbf_i(\xbf_{k+1})\|$, $\varepsilon_{k+1}$ and $\varepsilon_k$ in three cases:

\begin{enumerate}
    \item If $\| \gbf_i(\xkp) \| > \varepsilon_{k} > \varepsilon_{k+1}$, then by the definition in \eqref{eq:rei}, there is
\[
r_{\varepsilon_{k+1},i}(\xkp) + \frac{\varepsilon_{k+1}}{2} = \| \gbf_i(\xkp) \| = r_{\varepsilon_{k},i}(\xkp) + \frac{\varepsilon_{k}}{2}.
\]
    \item If $\varepsilon_{k} \ge \| \gbf_i(\xkp) \| > \varepsilon_{k+1}$, then \eqref{eq:rei} implies
\begin{align*}
r_{\varepsilon_{k+1},i}(\xkp) + \frac{\varepsilon_{k+1}}{2} = \| \gi(\xkp) \|
&= \frac{\| \gi(\xkp) \|}{2} + \frac{\| \gi(\xkp) \|}{2}
= \frac{\| \gi(\xkp) \|^2}{2\| \gi(\xkp) \|} + \frac{\| \gi(\xkp) \|}{2} \\
&\leqslant  \frac{\| \gi(\xkp) \|^2}{2\varepsilon_{k}} + \frac{\varepsilon_{k}}{2} = r_{\epsk,i}(\xkp) + \frac{\epsk}{2}.
\end{align*}
The second lines follows from the fact that $\frac{\|\gi(\xkp)\|^2}{2\varepsilon} + \frac{\varepsilon}{2}$---as a function of $\varepsilon$---is non-decreasing for all $\varepsilon \ge \| \gi(\xkp)\| $
    \item If $\varepsilon_{k}  > \varepsilon_{k+1} \ge \| \gbf_i(\xkp) \|$, then again the previous fact and \eqref{eq:rei} imply \eqref{eq:r_decay}.
\end{enumerate}

Therefore, in either of the three cases, \eqref{eq:r_decay} holds and hence
\[
r_{\varepsilon_{k+1}}(\xkp) + \frac{m\varepsilon_{k+1}}{2} = \sum_{i=1}^m \del[2]{ r_{\varepsilon_{k+1},i}(\xkp) + \frac{\varepsilon_{k+1}}{2}} \le \sum_{i=1}^m \del[2]{ r_{\epsk,i}(\xkp) + \frac{\epsk}{2} } = r_{\varepsilon_{k}}(\xkp) + \frac{m\varepsilon_{k}}{2},
\]
which implies the first inequality of \eqref{eq:phi_decay}.
\end{proof}

\begin{theorem} \label{theorem_convergence}
Suppose that $\{\xbf_{k}\}$ is the sequence generated by Algorithm \ref{alg:lda} with any initial $\xbf_0$. Let $\{\xbf_{\kl+1}\}$ be the subsequence where the reduction criterion $\varepsilon_{k+1}= \gamma {\varepsilon_k}$ in line 15 is met for $k=k_l$ and $l=1,2,\dots$. Then, $\{\xbf_{\kl + 1}\}$  has at least one accumulation point, and every accumulation point of $\{\xbf_{\kl + 1}\}$ is a Clarke stationary point of \eqref{eq:loa}.
\end{theorem}

\begin{proof}
By the Lemma 4.2 and the definition of Clarke subdifferential, this theorem can be easily proved in the similar way to Theorem 3.6 in \cite{chen2020learnable}.

Due to Lemma \ref{lem:phi_decay} and $\phi(\xbf) \le \phi_{\varepsilon}(\xbf) + \frac{m \varepsilon}{2}$ for all $\varepsilon>0$ and $\xbf \in \Xcal$, we know that
\begin{equation*}
\phi(\xk) \le \phi_{\epsk}(\xk) + \frac{m\epsk}{2} \le \cdots \le \phi_{\varepsilon_0}(\xbf_0) + \frac{m\varepsilon_0}{2} < \infty.
\end{equation*}
Since $\phi$ is coercive, we know that $\{\xk\}$ is bounded. Hence $\{ \xbf_{\kl+1}\}$ is also bounded and has at least one accumulation point.

Note that $\xbf_{\kl+1}$ satisfies the reduction criterion in Line 15 of Algorithm 1, i.e., $\| \nabla \phi_{\varepsilon_{\kl}} (\xbf_{\kl+1}) \| \le \sigma \gamma \varepsilon_{\kl} = \sigma \varepsilon_0 \gamma^{l+1} \to 0$ as $l \to \infty$.
For notation simplicity, let $\{\xjp\}$ denote any convergent subsequence of $\{\xbf_{\kl +1}\}$ and $\varepsilon_j$ the corresponding $\epsk$ used in the iteration to generate $\xjp$. Then there exists $\xhat \in \Xcal$ such that $\xjp \to \xhat$, $\epsj \to 0$, and $\nabla \phi_{\epsj}(\xjp) \to 0$ as $j\to \infty$.

Note that the Clarke subdifferential of $\phi$ at $\xhat$ is given by $\partial \phi(\xbf) = \partial \hat{r}(\xbf) + \lambda \nabla \overline{r}(\xbf) + \nabla f(\xbf)$:
\begin{equation}\label{eq:d_phi_xhat}
\partial \phi(\xhat) = \cbr[2]{\sum_{i \in I_0} \nabla \gi(\xhat)^{\top} \wbf_i + \sum_{ i \in I_1} \nabla \gi(\xhat)^{\top} \frac{\gi(\xhat)}{\| \gi(\xhat) \|} + \lambda \nabla \overline{r}(\xhat) + \nabla f(\xhat) \ \bigg\vert \ \| \Pi(\wbf_i; \Ccal(\nabla \gi(\xhat))) \le 1,\ \forall\, i\in I_0},
\end{equation}
where $I_0 = \{i\in[m]\ \vert \ \|\gi(\xhat)\| = 0 \}$ and $I_1 = [m] \setminus I_0$.
Then we know that there exists $J$ sufficiently large, such that
\[
\epsj < \frac{1}{2}\min \{ \|\gi(\xhat)\| \ \vert \ i\in I_1\} \le \frac{1}{2} \|\gi(\xhat)\| \le \| \gi(\xjp) \|, \quad \forall\, j\ge J,\quad \forall\, i\in I_1,
\]
where we used the facts that $\min \{ \|\gi(\xhat)\| \ \vert \ i\in I_1\}>0$ and $\epsj \to 0$ in the first inequality, and $\xjp \to \xhat$ and the continuity of $\gi$ for all $i$ in the last inequality.
Furthermore, we denote
\begin{equation*}
\sji := \begin{cases}
\frac{\gi(\xbf_{j+1})}{\epsj}, & \mbox{if}\ \|\gi(\xjp)\| \le \epsj, \\
\frac{\gi(\xbf_{j+1})}{\| \gi(\xjp) \|}, & \mbox{if}\ \|\gi(\xjp)\| > \epsj.
\end{cases}
\end{equation*}
Then we have
\begin{equation}\label{eq:d_phi_epsj}
\nabla \phi_{\epsj}(\xjp) = \sum_{i \in I_0} \nabla \gi(\xjp)^{\top} \sji + \sum_{ i \in I_1} \nabla \gi(\xjp)^{\top} \frac{\gi(\xjp)}{\| \gi(\xjp) \|} + \lambda \nabla \overline{r}(\xjp) + \nabla f(\xjp).
\end{equation}
Comparing \eqref{eq:d_phi_xhat} and \eqref{eq:d_phi_epsj}, we can see that the last two terms on the right hand side of \eqref{eq:d_phi_epsj} converge to those of \eqref{eq:d_phi_xhat}, respectively, due to the facts that $\xjp \to \xhat$ and the the continuity of $\gi,\nabla \gi, + \nabla \overline{r}, \nabla f$. Moreover, noting that $\|\Pi(\sji; \Ccal(\nabla \gi(\xhat)))\| \le \| \sji \| \le 1$, we can see that the first term on the right hand side of \eqref{eq:d_phi_epsj} also converges to the set formed by the first term of \eqref{eq:d_phi_xhat} due to the continuity of $\gi$ and $\nabla \gi$. Hence we know that
\[
\dist( \nabla \phi_{\epsj}(\xjp), \partial \phi(\xhat)) \to 0,
\]
as $j \to 0$. Since $\nabla \phi_{\epsj}(\xjp) \to 0$ and $\partial \phi(\xhat)$ is closed, we conclude that $0 \in \partial \phi(\xhat)$.
\end{proof}


From Theorem \ref{theorem_convergence}, we conclude that the output of our network converges to a (local) minimizer of the original regularized problem \eqref{eq:loa}.
\section{Experiments and Results}
\label{experiment}
Here we present our experiments on LDCT image reconstruction problems with various dose levels and compare with existing state-of-the-art algorithms in terms of image quality, run time and the number of parameters etc.
We adopt a warm start training strategy which imitates the iterating of optimization algorithm. More precisely, first we train the network with $K = 3$ phases, where each phase in the network corresponds to an iteration in optimization algorithm. After it converges, we add $2$ more phases and we continue training the $5$-phase network until it converges. We continue adding 2 more phases until there is no noticeable improvement.

As computing the similarity weight matrix $\mathcal{W}$ is high-cost in time and memory, we also experiment with approximation of $\mathcal{W}$ computed on the initial reconstruction $\xbf_0$ without updating in each iteration, i.e.
$
  \mathcal{W}_{ij} \approx \exp{\big(-\frac{\|\hat\gbf_i(\xbf_0) - \hat\gbf_j(\xbf_0)\|^2}{\delta^2}\big)}.
$
%
%
Thus $\overline{r}(\xbf)$ can be differentiated as
$
\nabla \overline{r}(\xbf)= 2 \cdot \sum_{q = 1}^{d} \nabla\hat\gbf^q(\xbf)^{\top}\mathcal{L} \ \hat\gbf^q(\xbf).
$
%
%
In the approximation scenario we compute $\mathcal{L}$ once and use it for all the phases, but in the other case we have an extra $\mathcal{O}(m^2)$ computation in each phase, every time we compute $\nabla \phi$. As shown in the Section \ref{sec: Ablation Study}, this approximation does not exacerbate the network performance much but can increase the running speed.

All the experiments are performed on a computer with Intel i7-6700K CPU at 3.40 GHz, 16 GB of memory, and a Nvidia GTX-1080Ti GPU of 11GB graphics card memory, and implemented with the PyTorch toolbox \cite{NEURIPS2019_9015} in Python.
The initial $\xbf_0$ is obtained by FBP algorithm. The spatial kernel size of the convolution and transposed convolution is set to be $3 \times 3$ and the channel number is set to $48$ with layer number $l=4$ as default.
The learned weights of convolutions and transposed convolutions are initialized by Xavier Initializer \cite{Glorot10understandingthe} and the starting $\varepsilon_0$ is initialized to be $0.001$. All the learnable parameters are trained by the Adam Optimizer with $\beta_1=0.9$ and $\beta_2=0.999$. The network is trained with learning rate 1e-4 for $200$ epochs when phase number $K = 3$, followed by $100$ epochs when adding more phases.

We test the performance of ELDA on the \enquote{\textit{2016 NIH-AAPM-Mayo Clinic Low-Dose CT Grand Challenge}} \cite{AAPM} which contains 5936 full-dose CT (FDCT) data from 10 patients, from which we randomly select $500$ images and resize them to the size $256 \times 256$. Then we randomly divide the dataset into $400$ images for training and $100$ for testing. Distance-driven algorithm \cite{de2002distance, de2004distance} is applied to
simulate the projections in fan-beam geometry.
The source-to-rotation center and detector-to-rotation center distances are both set to 250 mm. The physical image region covers 170 mm $\times$ 170 mm. On detector there are 512 detector elements each with width 0.72 mm. There are 1024 projection views in total with the projection angles are evenly distributed over a full scan range.
Similar to \cite{possion_noise_and_gussian}, the simulated noisy transmission measurement $I$ is generated by adding Poisson and electronic noise as
\begin{equation}\label{eq:noise}
I = Possion(I_0 \exp{(-\hat{b})}) + Normal(0, \sigma_e^2),
\end{equation}
where $I_0$ is the incident X-ray intensity and $\sigma_e^2$ is the variance the background electronic noise. And $\hat{b}$ represents the noise-free projection. In this simulation, $I_0$ is set to $1.0 \times 10^6$ for normal dose and $\sigma_e^2$ is prefixed to be 10 for all dose cases.
Then the noisy projection $\bbf$ is calculated by taking the logarithm transformation on $\frac{I_0}{I}$.
In low dose case, in order to investigate the robustness of all compared algorithms, we generated three sets of different low dose projections with $I_0 = 1.0 \times 10 ^5, 5.0 \times 10 ^4 $ and $2.5 \times 10 ^4$ which correspond to $10 \%$, $5 \%$ and $2.5 \%$ of the full dose incident accordingly.
We use peak signal to noise ratio (PSNR) and structureal similarity index measure (SSIM) to evaluate the quality of the reconstructed images.
\subsection{Parameter Study}
\label{sec: Parameter Study}
The regularization term of our model is learned from training samples, yet there are still a few key network hyperparameters need to be set manually. Specifically, we investigate the impacts of some parameters of the architecture, which includes the number of convolutions ($l$), the depth of the convolution kernels ($d$) and the phase number ($K$). The influence of each hyperparameter is sensed by perturbing it with others fixed at $d = 48$, $l = 4$ and $K = 19$. The setting includes all factors/components listed in Table. \ref{tab:Components} and all following results are trained and tested with dose level $10 \%$.

 \textbf{\textit{Depth of the convolution kernels:}} We evaluate the instances of $d = 16, 32, 48$ and $64$. The results are listed in Table. \ref{tab:depth}.
 It is evident that the PSNR score raises with growing depth of the kernels, but the profit gradually reduces. On the contrary, the number of parameters and running time grow greatly in the meantime. %
%
\begin{table}[h]
\centering
\addtolength{\tabcolsep}{-1pt}
\caption{
The reconstruction results associated with different depths of convolution kernels and different number of convolutions in each phase with dose level 10\%.
}
\begin{tabular}{lcccc}
\toprule
\textbf{Depth of conv. kernels} & \textbf{16} & \textbf{32} & \textbf{48} & \textbf{64}\\
\midrule
PSNR (dB) & 47.11 & 47.51 & 47.73 & 47.75 \\
Number of parameters & 14,152 & 55,912 & 125,320 & 222,376 \\
Average testing time (s) & 1.139 & 1.231 & 1.411 & 1.539 \\
\bottomrule
\toprule
\textbf{Number of convolutions} & \textbf{2} & \textbf{3} & \textbf{4} & \textbf{5}\\
\midrule
PSNR (dB) & 47.27 & 47.60 & 47.73 & 47.77\\
Number of parameters & 42,376 & 83,848 & 125,320 & 167,656\\
Average testing time (s) & 1.117 & 1.258 & 1.411 & 1.530 \\
\bottomrule
\end{tabular}
\label{tab:depth}
\end{table}

\begin{figure}[h]
\centering
\includegraphics[width=0.475\textwidth]{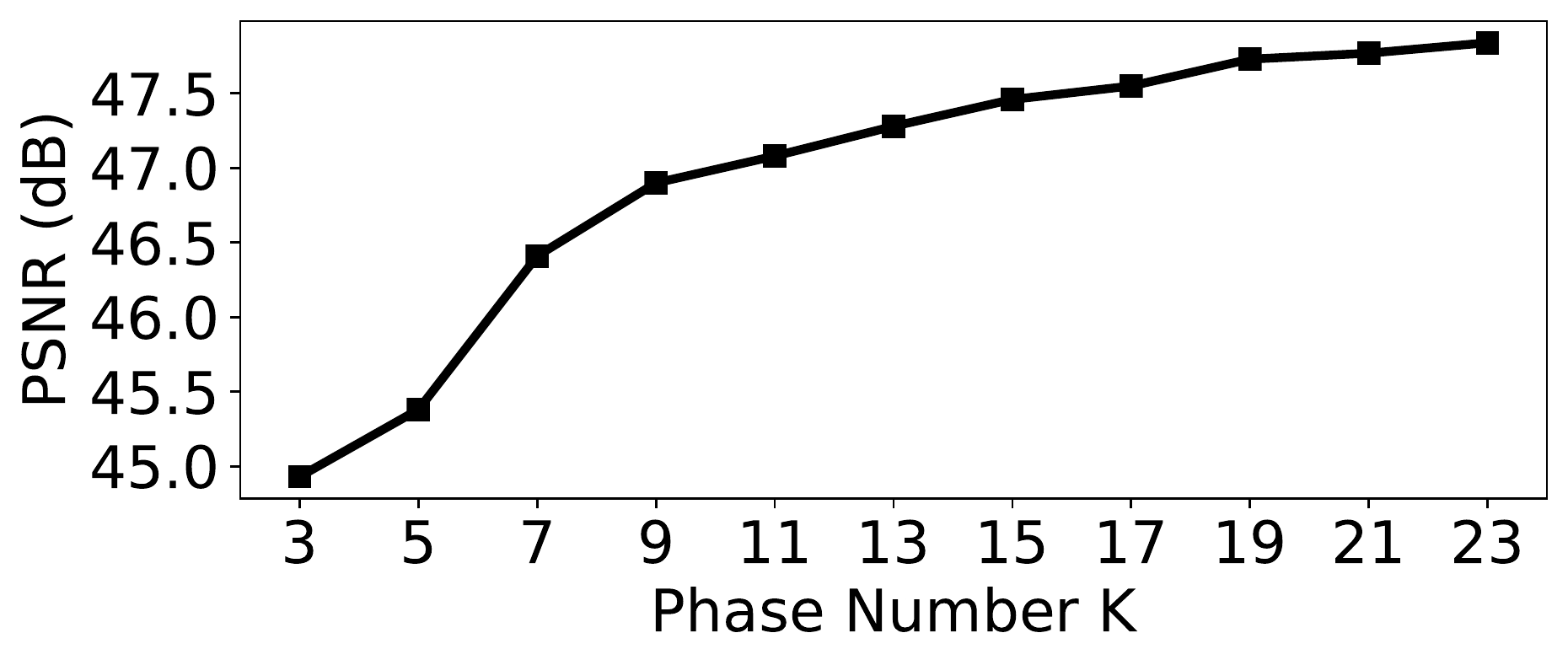}
\caption{The reconstruction PSNR of ELDA across various phase numbers on dose level $10\%$.}
\label{psnr-phase}
\end{figure}
\begin{table}[h]
\centering
\addtolength{\tabcolsep}{-1.25pt}
\caption{
Comparison of different algorithms and the influence of various design factors or components on the LDCT reconstruction performance with dose level 10\%.
}
\begin{tabular}{lcccc}
\toprule
\textbf{Algorithms} & \textbf{Plain-GD} & \textbf{AGD} & \textbf{LDA} & \textbf{ELDA}\\
\midrule
PSNR (dB) & 45.47 & 45.95 & 46.29 & 46.35\\
Average testing time (s) & 0.602 & 0.605 & 1.338 & 1.239 \\
\bottomrule
\toprule
\multicolumn{5}{c}{\textbf{ELDA Factors / Components}} \\
\midrule
Inexact Transpose? &\xmark & \cmark& \cmark& \cmark\\
Non-local regularizer?  & \xmark & \xmark & \cmark& \cmark\\
Approximated Matrix $\mathcal{W}$? & \xmark & \xmark & \xmark& \cmark\\
\midrule
PSNR (dB)& 46.35 & 46.74 & 47.75 & 47.73\\
Average testing time (s) & 1.239 & 1.247 & 3.703 & 1.411 \\
\bottomrule
\end{tabular}
\label{tab:Components}
\end{table}
\textbf{\textit{Number of convolutions:}} We evaluate the cases of different number of convolutions $l = 2, 3, 4$ and $5$. The corresponding results are reported in Table. \ref{tab:depth}. We can observe that more convolutions contribute to better reconstructed image quality. But the increase of PSNR score is insignificant from $4$ convolutions to $5$ while the parameter number and the test time rise significantly. 

\textbf{\textit{Phase number $\mathbf{K}$:}}
As shown in Fig. \ref{psnr-phase}, the PSNR increases with the phase number $K$. And the plot approaches nearly flat after 19 phases.

To balance the trade-off between reconstruction performance and network complexity, we take $d = 48$, $l = 4$, $K = 19$ when comparing with other methods.
\begin{table*}[h]
\centering
\caption{Quantitative results (Mean $\pm$ Standard Deviation) of the LDCT reconstructions of AAPM-Mayo data obtained by various algorithms and different dose levels. Average run time (per image) is measured in second.}
\label{www}
\small
\addtolength{\tabcolsep}{-6.0pt}
\begin{tabular}{ccccccccc}
\toprule
\multirow{2}{*}{\textbf{Dose Level}} & \multicolumn{2}{c}{\textbf{$1.0 \times 10^5$}} & \multicolumn{2}{c}{\textbf{$5.0 \times 10^4$}} & \multicolumn{2}{c}{\textbf{$2.5 \times 10^4$}} & \multirow{2}{*}{Run Time}& \multirow{2}{*}{Parameters}\\
& PSNR (dB) &  SSIM & PSNR (dB) & SSIM & PSNR (dB) & SSIM & & \\
\midrule
FBP \cite{kak2002principles}& 38.03$\pm$0.68  & 0.9084$\pm$0.0128 &35.25$\pm$0.74 & 0.8459$\pm$0.0209 &32.35$\pm$0.77 &0.7556$\pm$0.0295 & 1.0 & N/A\\
TGV \cite{niu2014sparse} & 43.60$\pm$0.50 & 0.9845$\pm$0.0020& 42.80$\pm$0.52 & 0.9812$\pm$0.0027 & 41.10$\pm$0.62 & 0.9698$\pm$0.0055 & $9.4 \cdot 10 ^{2}$ & N/A\\
FBPConvNet \cite{jin2017deep} & 42.49$\pm$0.47 & 0.9764$\pm$0.0029 & 41.14$\pm$0.49 & 0.9698$\pm$0.0041 & 39.62$\pm$0.48 & 0.9589$\pm$0.0052 & $7.4 \cdot 10^{-2}$ & $1.0 \cdot 10 ^7$\\
RED-CNN \cite{CNN4} & 44.58$\pm$0.58 &0.9848$\pm$0.0026 &43.25$\pm$0.60 & 0.9808$\pm$0.0033 & 41.61$\pm$0.63 &0.9737$\pm$0.0048 &$\mathbf{2.6 \cdot 10 ^{-3}}$ & $1.8 \cdot 10 ^6$\\
Learned PD \cite{adler2018learned} & 44.81$\pm$0.59 &0.9863$\pm$0.0025 &43.28$\pm$0.61 & 0.9813$\pm$0.0034 & 41.74$\pm$0.62 & 0.9745$\pm$0.0048 &1.8 & $2.5 \cdot 10 ^5$\\
LEARN \cite{chen2018learn}& 45.19$\pm$0.60 & 0.9868$\pm$0.0023 &43.86$\pm$0.61 & 0.9840$\pm$0.0027 & 42.06$\pm$0.61 &0.9776$\pm$0.0037 &4.8 & $1.9 \cdot 10 ^6$\\
LDA \cite{chen2020learnable}& 46.29$\pm$0.60 &0.9896$\pm$0.0019 & 44.64$\pm$0.61 &0.9858$\pm$0.0025 & 42.97$\pm$0.64 &0.9802$\pm$0.0035 & 1.3& $\mathbf{6.2 \cdot 10 ^4}$\\
MAGIC \cite{xia2020magic} & 47.54$\pm$0.66 & 0.9919$\pm$0.0017 & 45.83$\pm$0.64 &0.9888$\pm$0.0022 & 44.37$\pm$0.63 & 0.9853$\pm$0.0027 &5.3 & $2.1 \cdot 10 ^6$\\
\textbf{ELDA (Ours)} & \textbf{47.73$\pm$0.69} & \textbf{0.9923$\pm$0.0017}& \textbf{46.40$\pm$0.64}& \textbf{0.9899$\pm$0.0021} & \textbf{44.86$\pm$0.65 }&\textbf{0.9859$\pm$0.0031} & 1.4&$1.2 \cdot 10 ^5$ \\
\bottomrule
\end{tabular}
\end{table*}
\begin{figure*}[h]
\centering
\subfigure[Reference ]{\includegraphics[width=0.1925\textwidth]{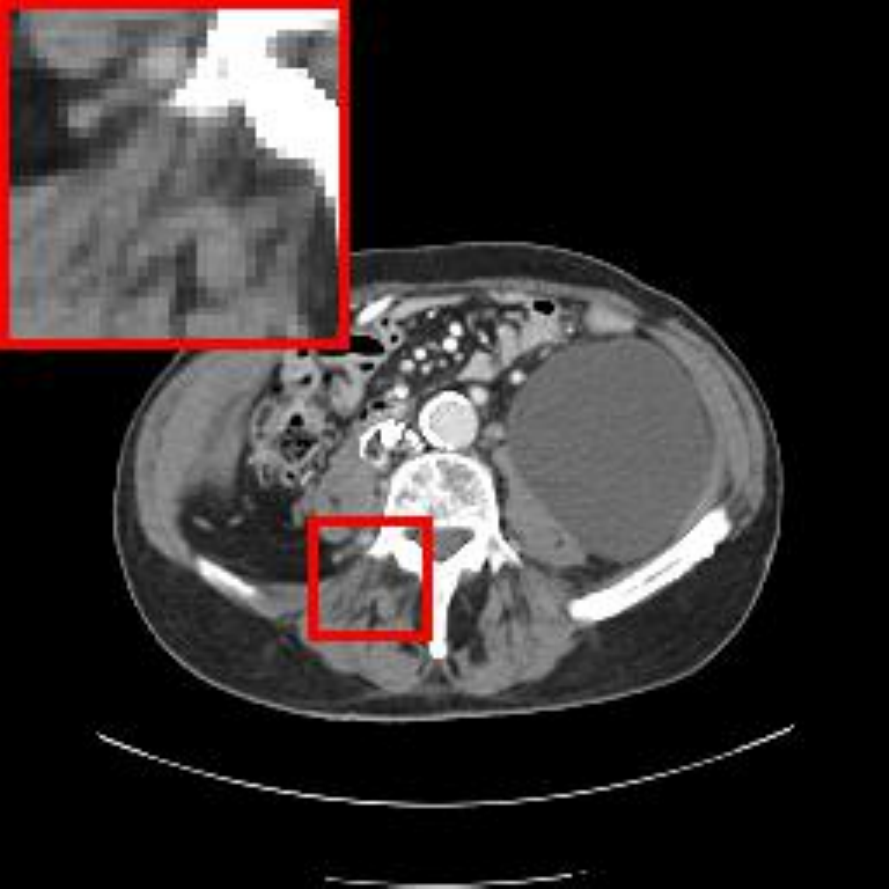}}
\subfigure[FBP (36.20)]{\includegraphics[width=0.1925\textwidth]{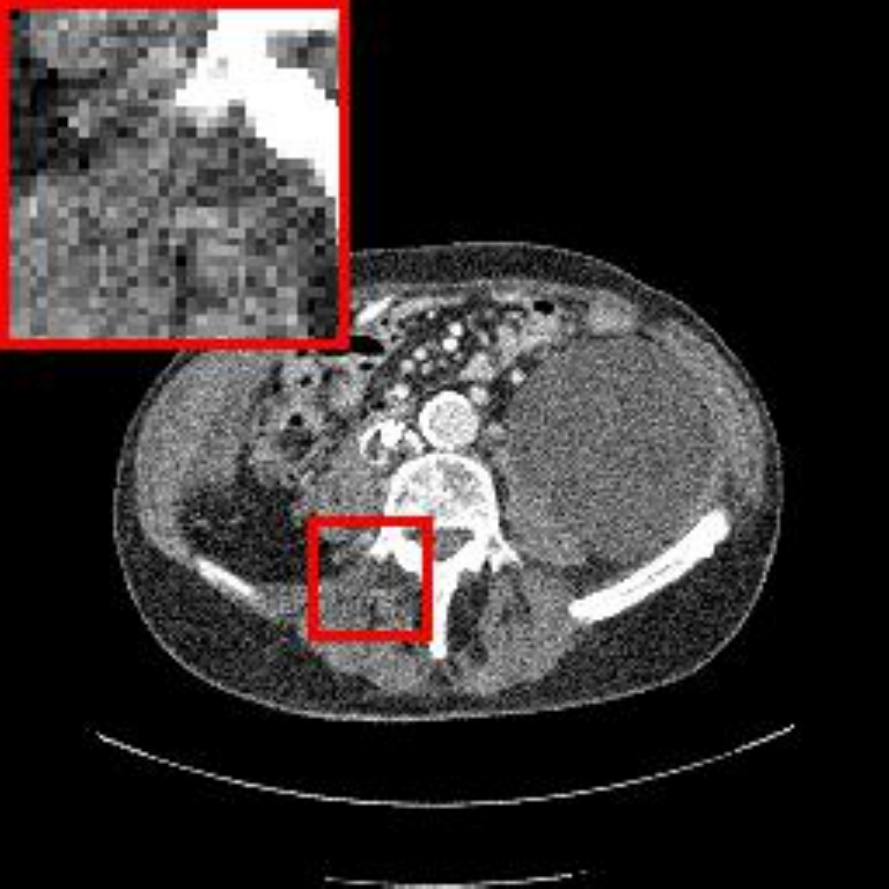}}
\subfigure[TGV (43.40)]{\includegraphics[width=0.1925\textwidth]{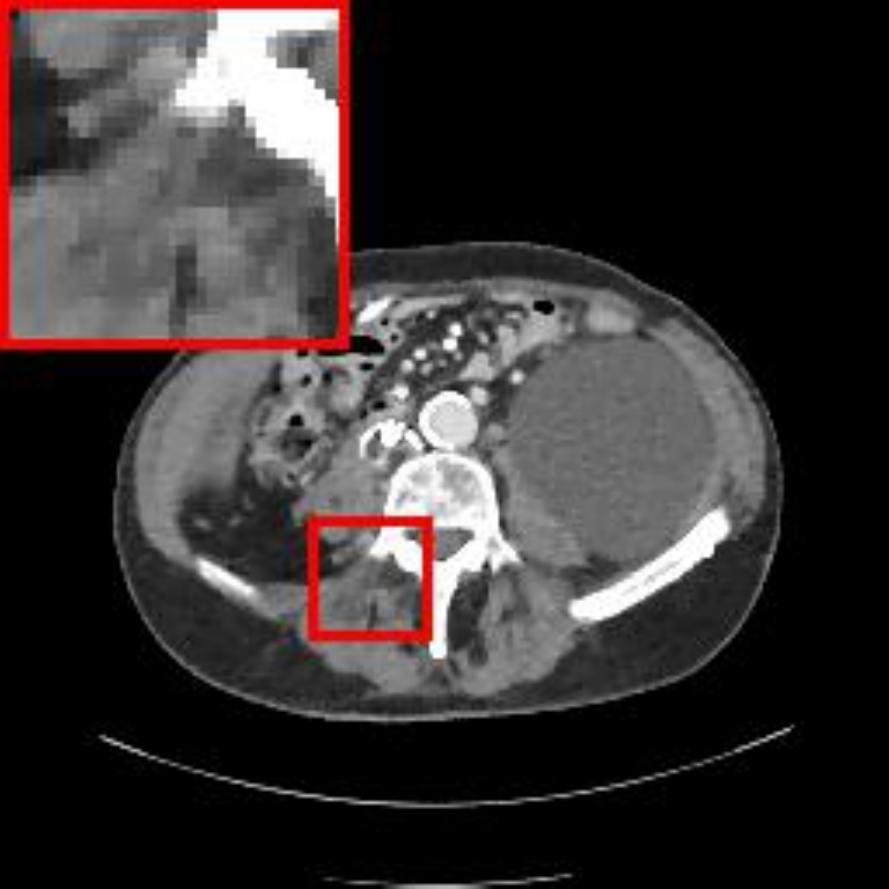}}
\subfigure[FBPConvNet (41.65)]{\includegraphics[width=0.1925\textwidth]{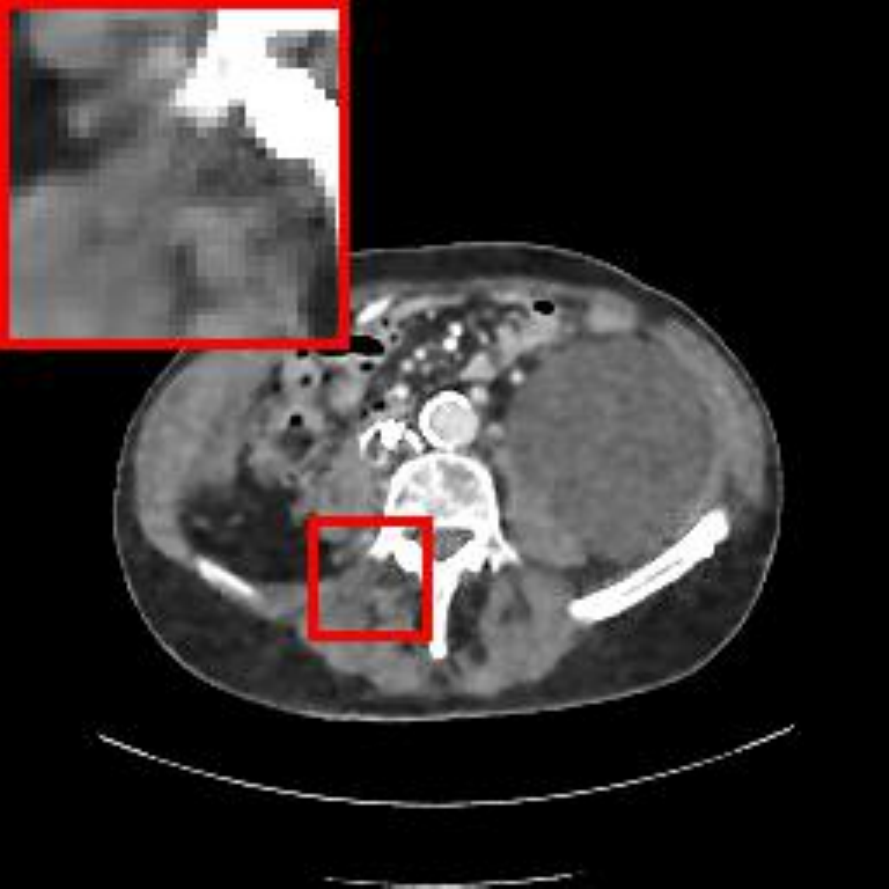}}
\subfigure[RED-CNN (43.84)]{\includegraphics[width=0.1925\textwidth]{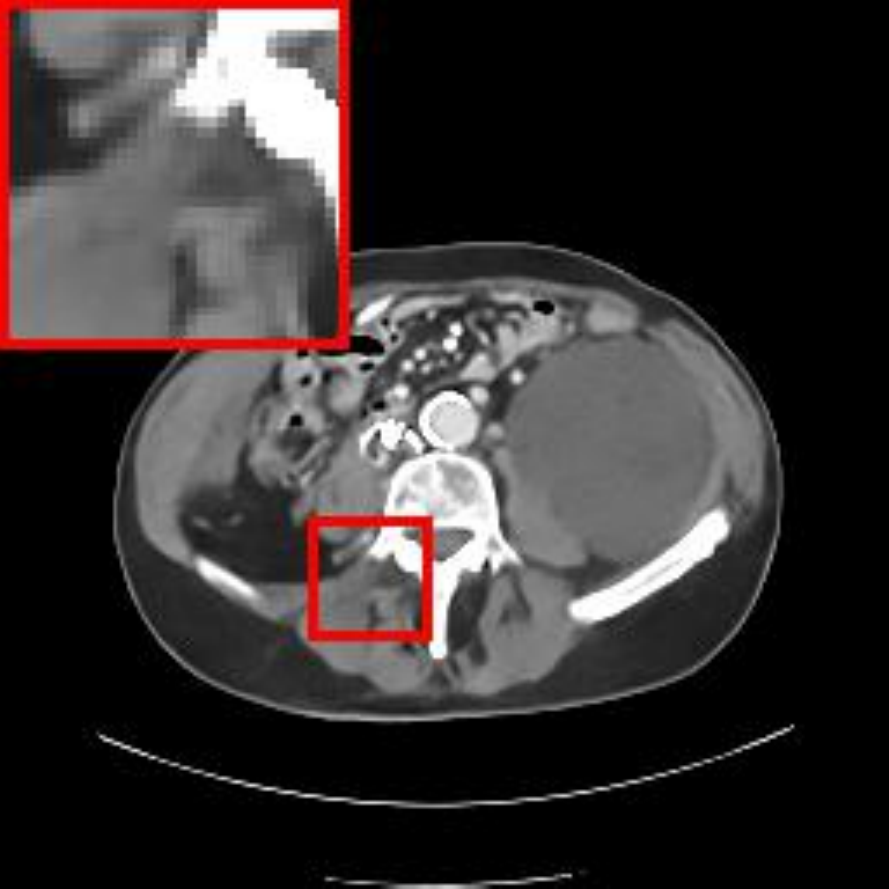}}
\subfigure[Learned PD (43.90)]{\includegraphics[width=0.1925\textwidth]{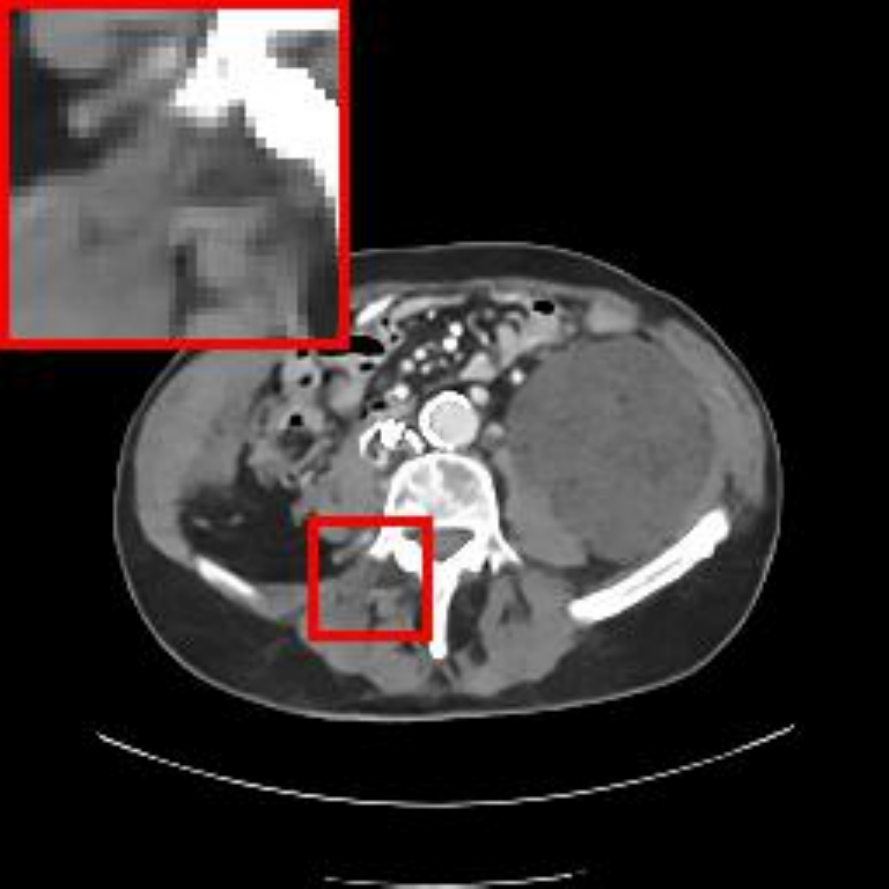}}
\subfigure[LEARN (44.31)]{\includegraphics[width=0.1925\textwidth]{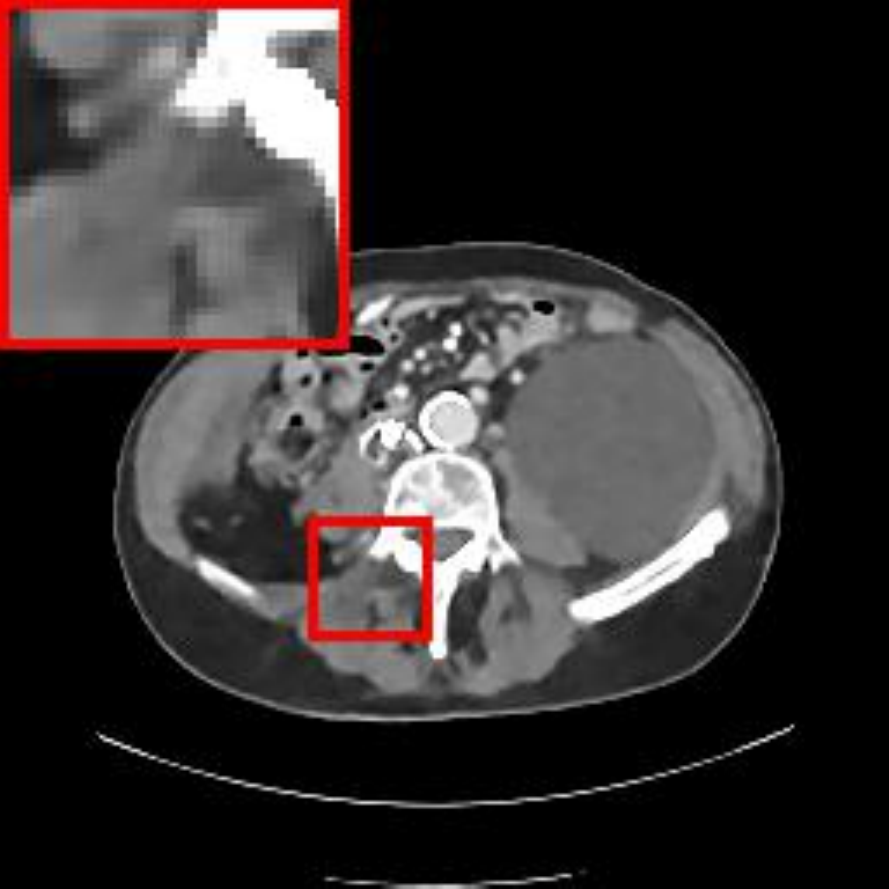}}
\subfigure[LDA (45.01)]{\includegraphics[width=0.1925\textwidth]{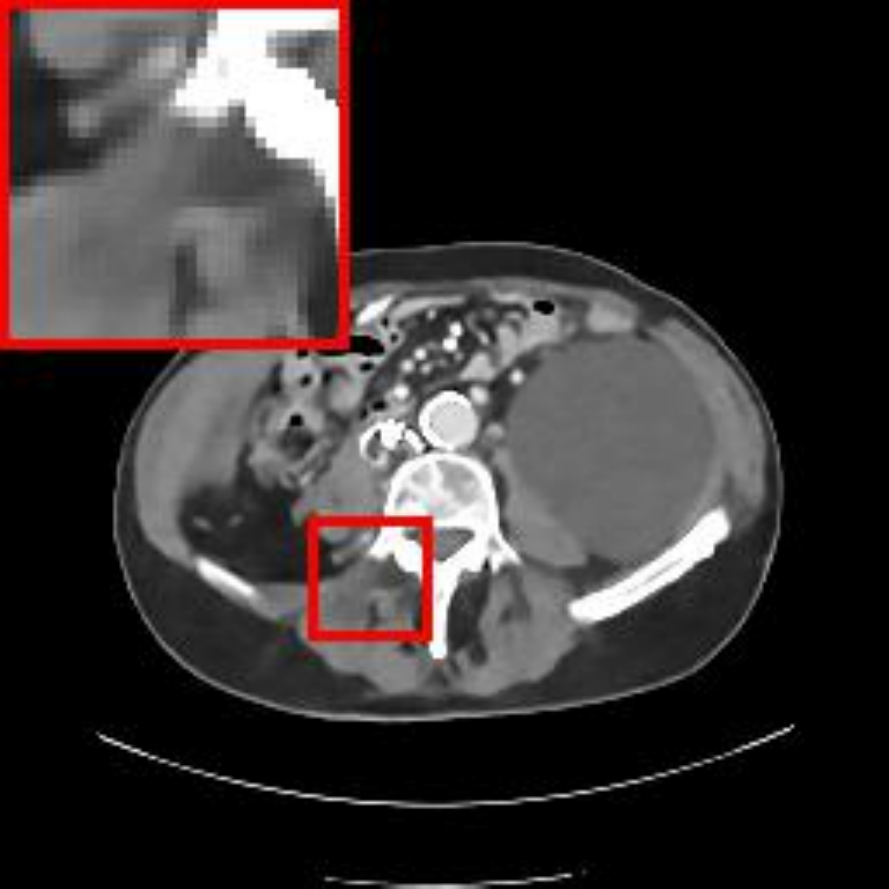}}
\subfigure[MAGIC (46.42)]{\includegraphics[width=0.1925\textwidth]{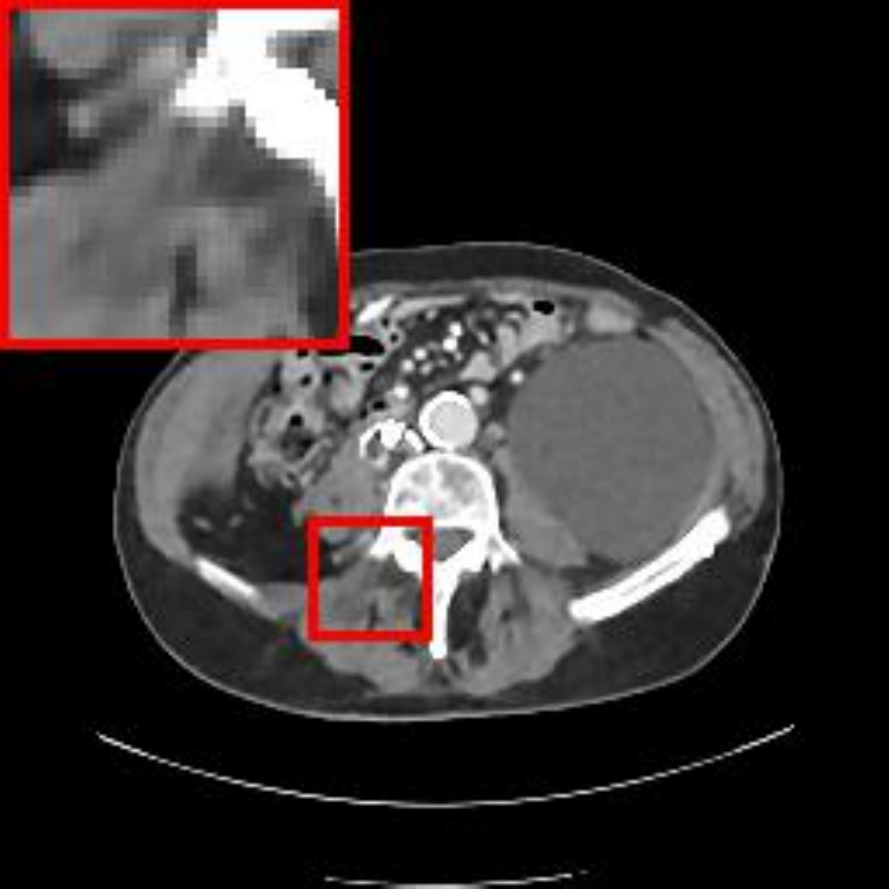}}
\subfigure[ELDA (47.14)]{\includegraphics[width=0.1925\textwidth]{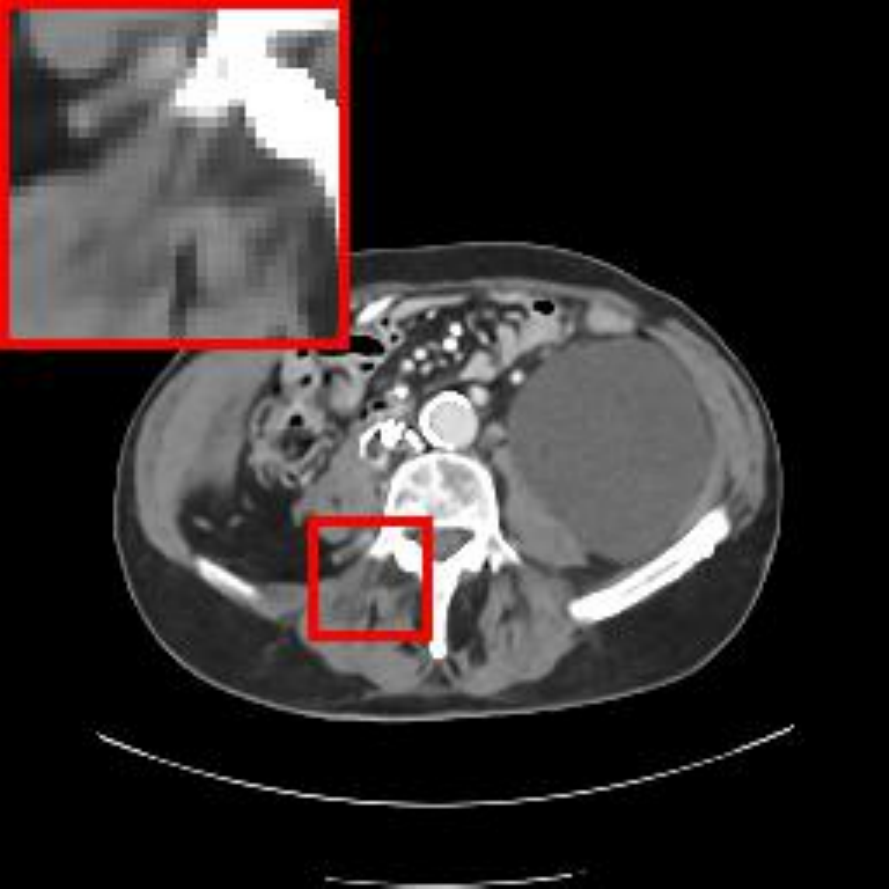}}
\caption{
Representative CT images of AAPM-Mayo data reconstructed by various methods with dose level 5\%. The display window is [-160, 240] HU. PSNRs (dB) are shown in the parentheses. The regions of interest are magnified in red boxes for better visualization.}
\label{fig:ct_rec_aapm}
\end{figure*}
\subsection{Ablation Study}
\label{sec: Ablation Study}
In this section, we first investigate the effectiveness of the proposed algorithm in ELDA. To this end, we compare ELDA with unrolling the standard gradient descent iteration of \eqref{eq:loa}, and an accelerated inertial version by setting $\xbf_{k+1} = \xbf_k - \alpha_k \nabla\phi(\xbf_k) + \theta_k (\xbf_k - \xbf_{k-1})$ where $\theta_k$ is also learned. Here these two algorithms are named as Plain-GD and AGD, respectively. In addition, to show the
superiority of the new descending condition \eqref{condition:u}\&\eqref{condition:v} over the competition strategy in LDA \cite{chen2020learnable}, we compare the result with LDA here as well. The comparison of different algorithms are shown in Table. \ref{tab:Components} (up), where the experiments follow the default parameter configuration as Section \ref{sec: Parameter Study} without the additional components listed in Table. \ref{tab:Components} (bottom). It is quite obvious that all AGD, LDA and ELDA achieve higher PNSRs than Plain-GD, where ELDA achieves the best. With the new descending condition it achieves average 0.06 dB PSNR better than LDA and about 0.1 seconds faster.
Furthermore, we have also computed the ratios that the candidate $\ukp$ is taken instead of $\vkp$, and found it to be $86.2\%$ for LDA and $99.6\%$ for ELDA respectively. That indicates that the proposed descending condition \eqref{condition:u} can effectively avoid the frequent candidate alternating compared to the competition strategy used in LDA.

Moreover, we check the influence of some essential factors/components of our ELDA model, i.e. the inexact transpose, the nonlocal smoothing regularizer and the approximated weight matrix $\mathcal{W}$. The results are summarized in Table. \ref{tab:Components} (bottom).
%
%
It is remarkable that the inexact transpose and the nonlocal smoothing regularizer can effectively increase the network performance by a large margin. And with the approximated weight matrix $\mathcal{W}$ there is no significant decreasing of the PSNR. As the initial $\xbf_0$ obtained by FBP is not far from $\xbf_k$ in each iteration, the $\mathcal{W}$ approximated by $\xbf_0$ can provide a good estimation to the true one. Thus, in the following sections we will keep all these features in Table \ref{tab:Components} (bottom) when comparing ELDA with other methods.
\begin{figure*}[h]
\centering
\subfigure[Reference ]{\includegraphics[width=0.1925\textwidth]{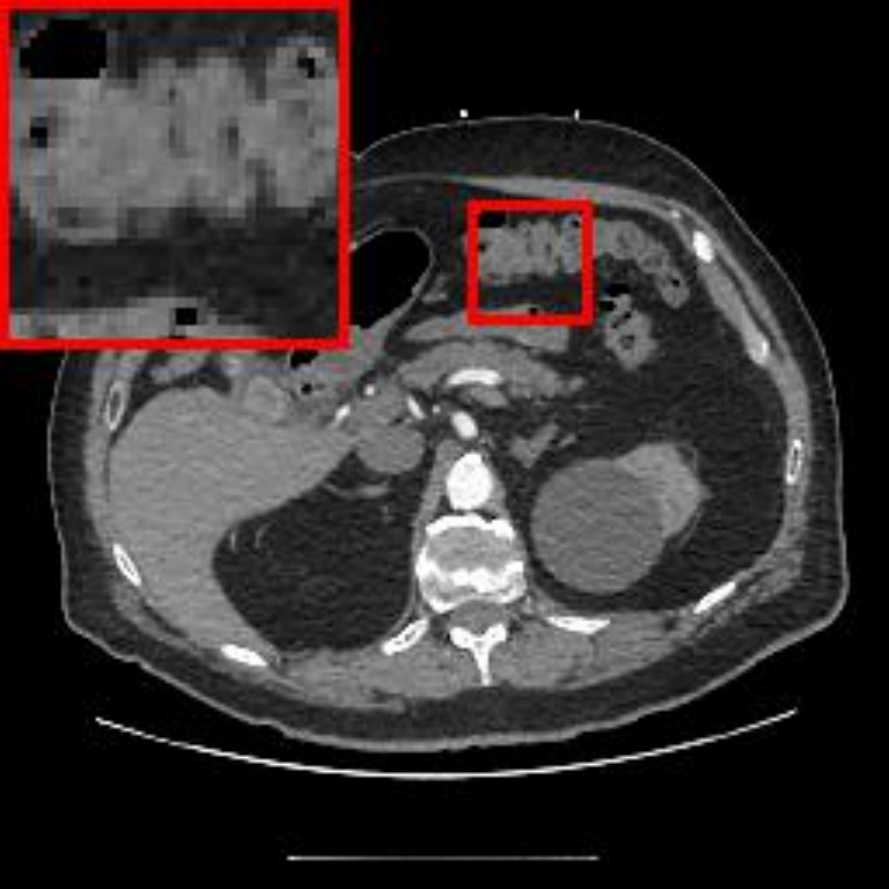}}
\subfigure[FBP (31.09)]{\includegraphics[width=0.1925\textwidth]{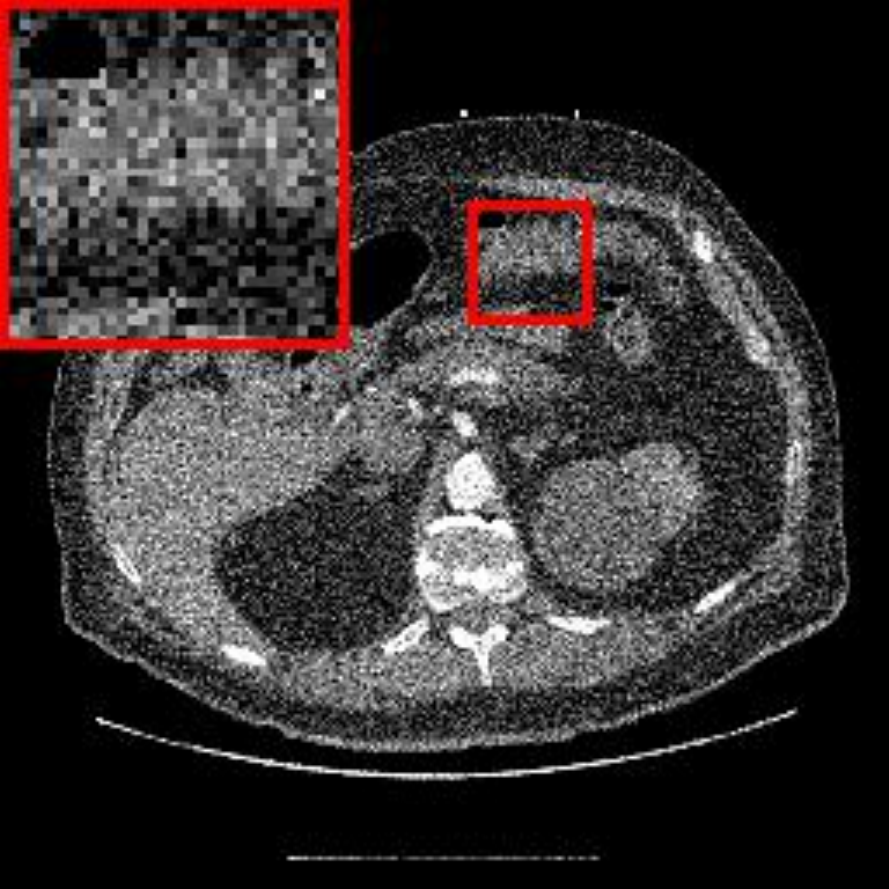}}
\subfigure[TGV (39.65)]{\includegraphics[width=0.1925\textwidth]{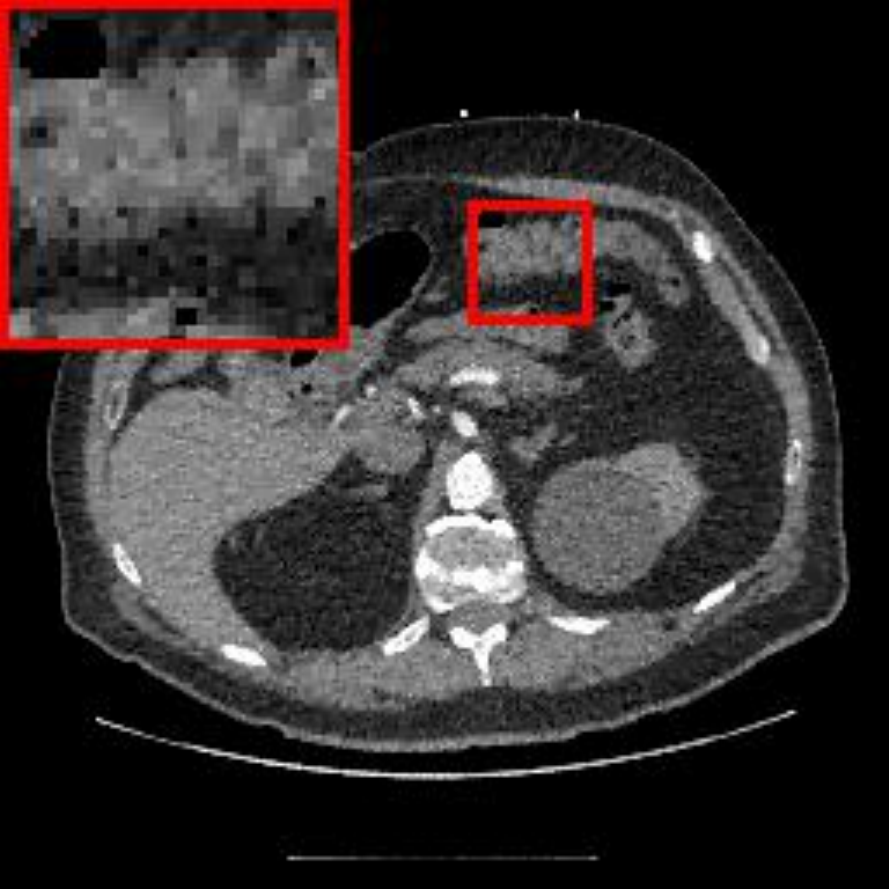}}
\subfigure[FBPConvNet (38.27)]{\includegraphics[width=0.1925\textwidth]{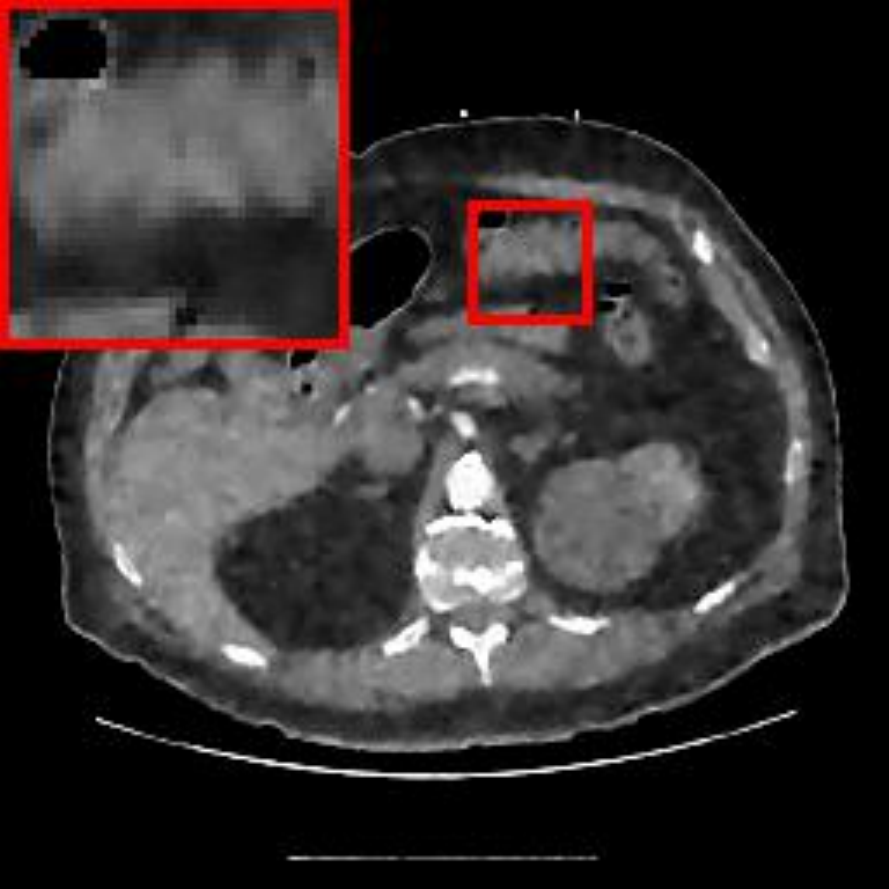}}
\subfigure[RED-CNN (39.49)]{\includegraphics[width=0.1925\textwidth]{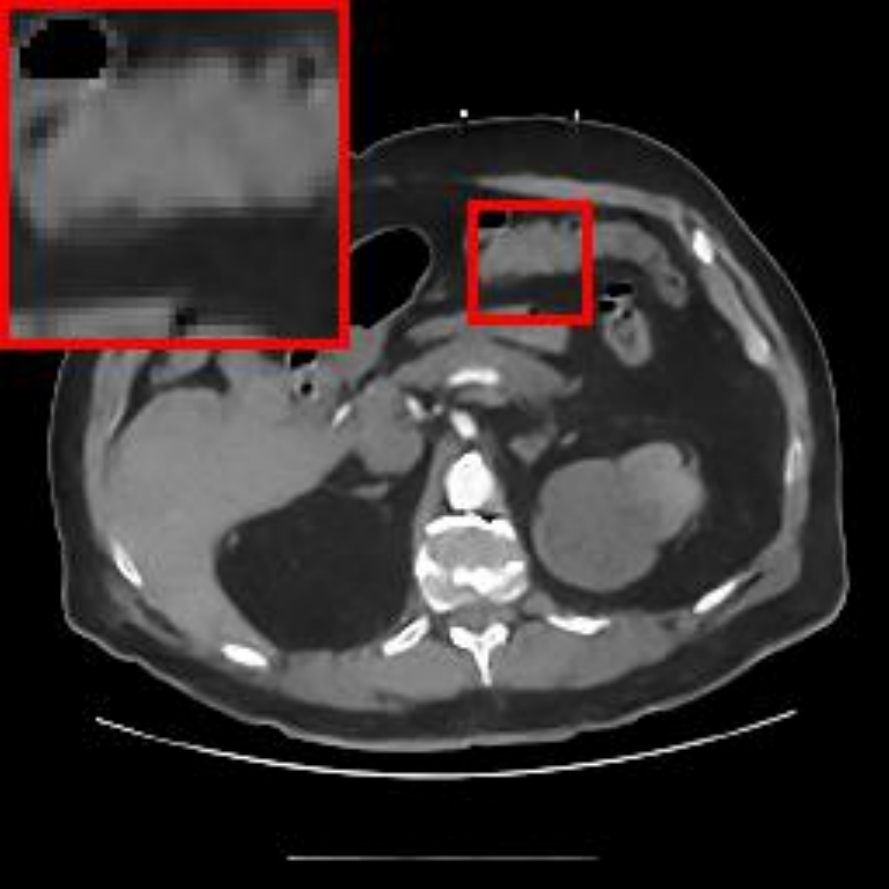}}
\subfigure[Learned PD (39.60)]{\includegraphics[width=0.1925\textwidth]{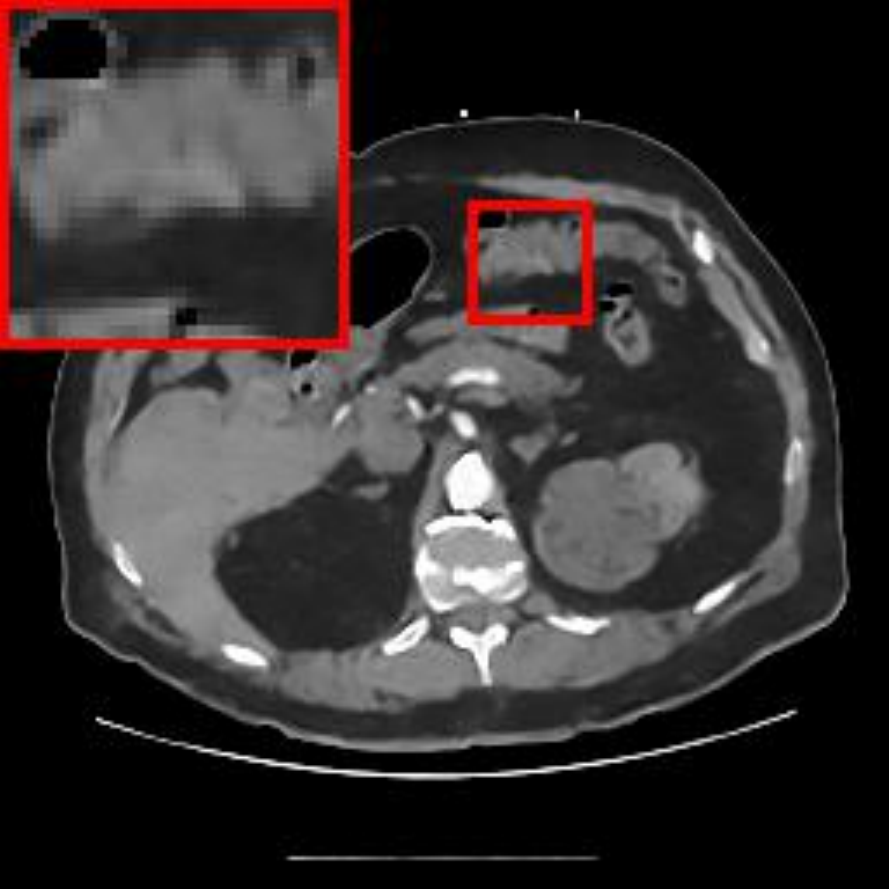}}
\subfigure[LEARN (40.11)]{\includegraphics[width=0.1925\textwidth]{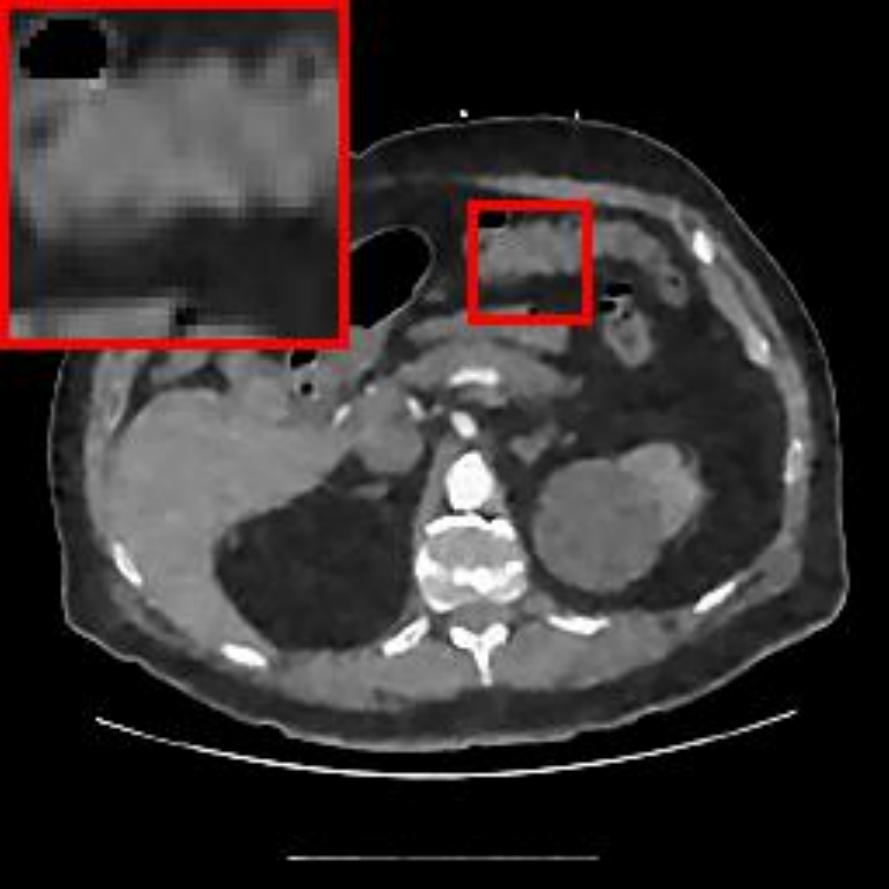}}
\subfigure[LDA (40.79)]{\includegraphics[width=0.1925\textwidth]{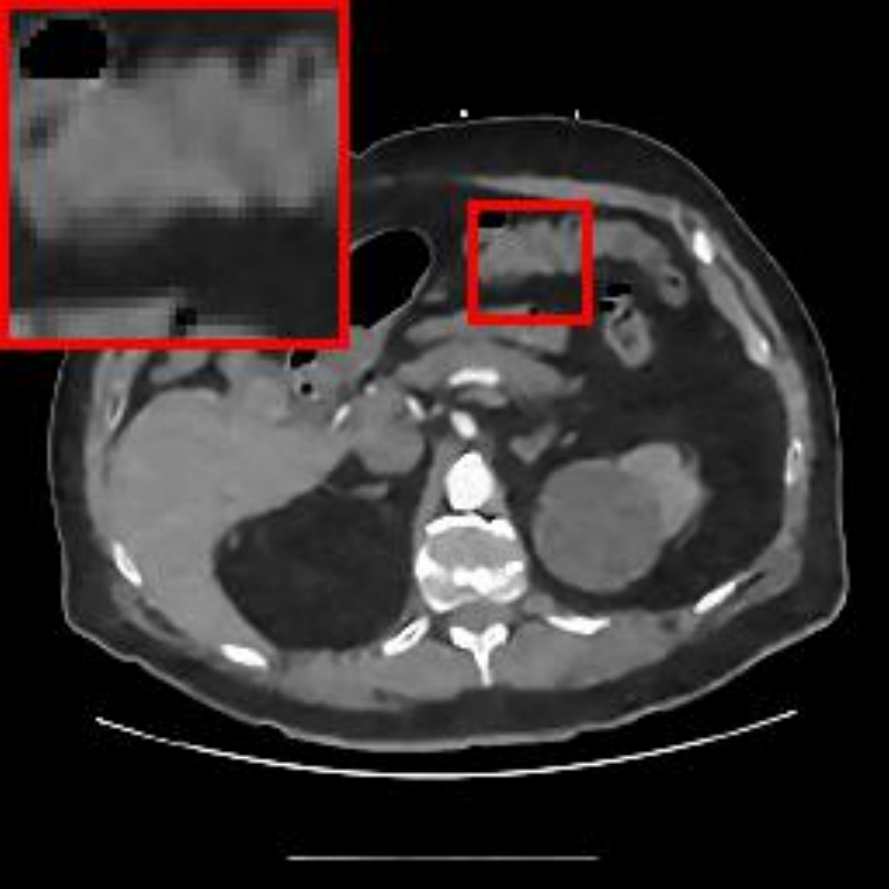}}
\subfigure[MAGIC (41.27)]{\includegraphics[width=0.1925\textwidth]{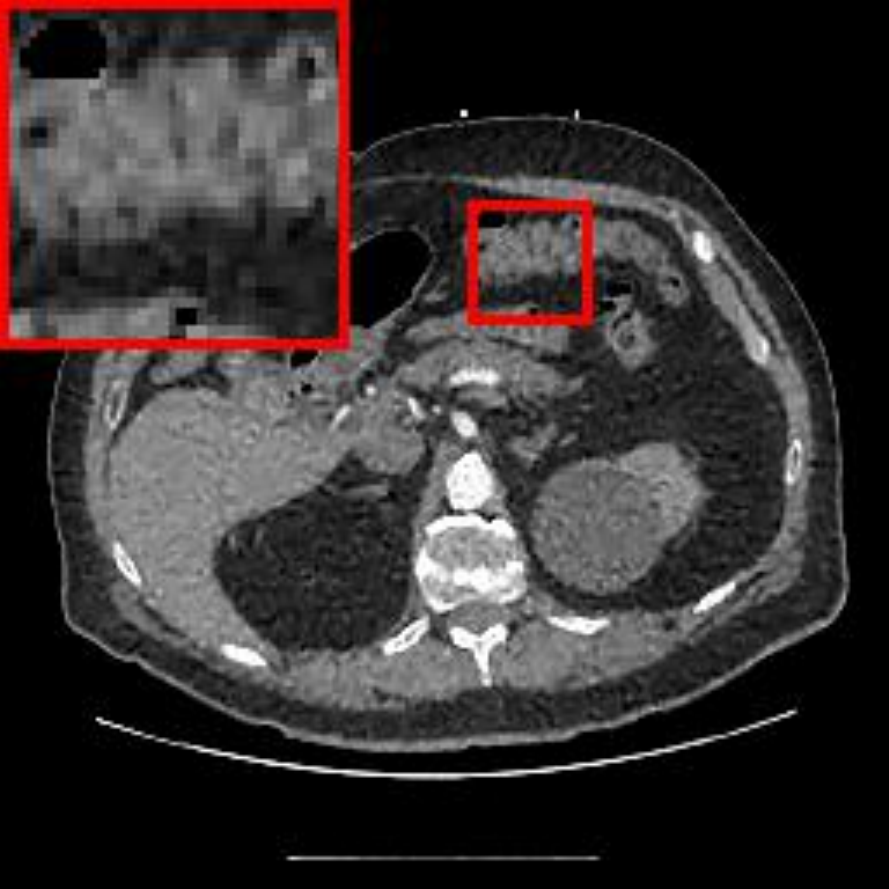}}
\subfigure[ELDA (42.41)]{\includegraphics[width=0.1925\textwidth]{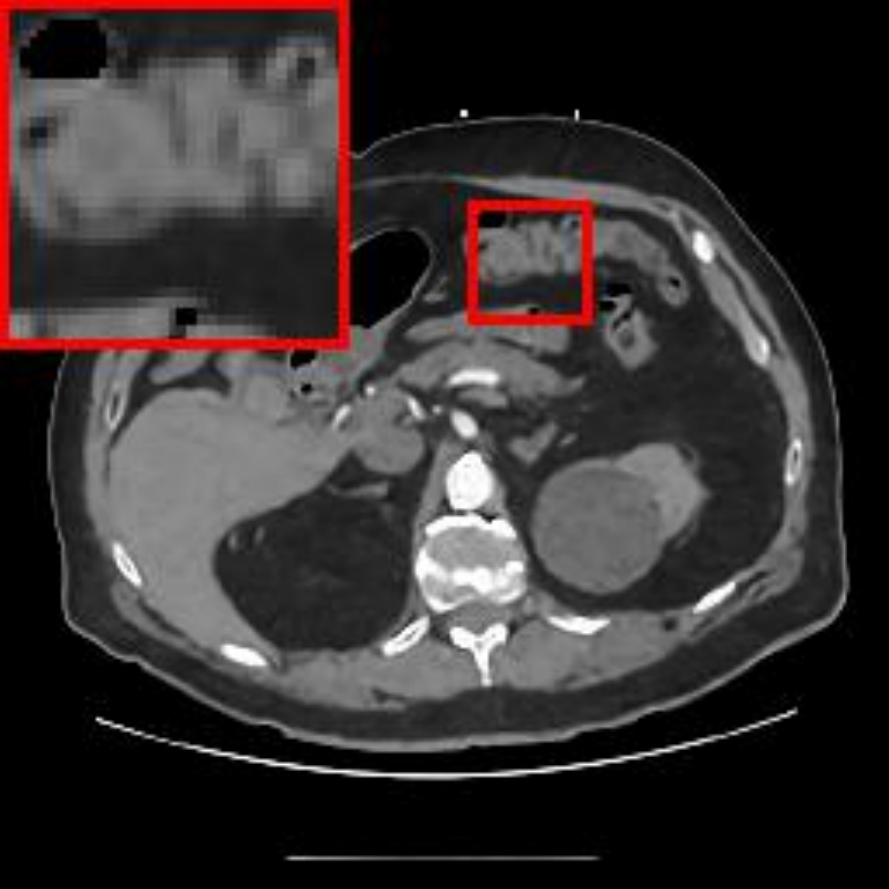}}
\caption{
Representative CT images of NBIA data reconstructed by various methods with dose level 2.5\%. The display window is [-160, 240] HU. PSNRs (dB) are shown in the parentheses. The regions of interest are magnified in red boxes for better visualization.}
\label{fig:ct_rec_}
\end{figure*}

\subsection{ Comparison with the State-of-the-art}
\label{ Comparison with the State-of-the-art}
In this section, we compare our reconstruction results on the 100 AAPM-Mayo testing images with several existing algorithms: two classic reconstruction methods, i.e., FBP \cite{kak2002principles} and TGV \cite{niu2014sparse} and six approaches based on deep learning, i.e., FBPConvNet \cite{jin2017deep}, RED-CNN \cite{CNN4}, Learned Primal-Dual \cite{adler2018learned}, LDA \cite{chen2020learnable}, LEARN \cite{chen2018learn} and MAGIC \cite{xia2020magic}. For fair comparison, all deep learning algorithms compared here are trained and evaluated on the same dataset, dose levels and evaluation metrics. The experimental results on various dose levels  are summarized in Table. \ref{www} and the representative qualitative results on dose level $5\%$ are shown in Fig. \ref{fig:ct_rec_aapm}.
These results show that ELDA reconstructs more accurate images using relatively much fewer network parameters and decent running time.
\label{sec: experiments_simulated_data}
\subsection{Validation with NBIA Data}
\label{sec: experiments_NBIA_data}
To demonstrate the generalizability of the proposed method, we validate our model on another dataset of the National Biomedical Imaging Archive (NBIA). We randomly sampled 80 images from the NBIA dataset with various parts of the human body for diversity. For fair comparison, all deep learning based reconstruction models compared here are trained on the same dataset identical to Section \ref{ Comparison with the State-of-the-art}. Fig. \ref{fig:ct_rec_} visualizes the reconstructed images obtained by different methods under dose level 2.5\%.
It can be seen that ELDA preserves the details well, avoids over-smoothing and reduces artifacts, which gives the promising reconstruction quality in Fig. \ref{fig:ct_rec_}. The quantitative results are provided in Table. \ref{table_result_2}.
\begin{table}[H]
\centering
\caption{Quantitative Results (Mean $\pm$ Standard Deviation) of the LDCT Reconstructions Obtained by Various Algorithms and Different Dose Levels on NBIA Data.}
\label{www}
\addtolength{\tabcolsep}{-6.0pt}
\begin{tabular}{ccccccc}
\toprule
\multirow{2}{*}{\textbf{Dose Level}} & \multicolumn{2}{c}{\textbf{$1.0 \times 10^5$}} & \multicolumn{2}{c}{\textbf{$5.0 \times 10^4$}} & \multicolumn{2}{c}{\textbf{$2.5 \times 10^4$}}\\
& PSNR (dB) &  SSIM & PSNR (dB) & SSIM & PSNR (dB) & SSIM \\
\midrule
FBP& 37.36$\pm$0.43 &0.9125$\pm$0.0145 &34.61$\pm$0.51 & 0.8533$\pm$0.0240 &31.74$\pm$0.55 &0.7682$\pm$0.0339 \\
FBPConvNet & 40.86$\pm$0.25 & 0.9693$\pm$0.0026 & 39.27$\pm$0.29 & 0.9587$\pm$0.0033 & 37.69$\pm$0.38 & 0.9449$\pm$0.0036 \\
RED-CNN & 41.74$\pm$0.32 &0.9751$\pm$0.0022 &40.32$\pm$0.38 & 0.9683$\pm$0.0026 & 38.76$\pm$0.46  &0.9578$\pm$0.0029 \\
Learned PD & 41.86$\pm$0.31 &0.9760$\pm$0.0025
 &40.39$\pm$0.39 &0.9687$\pm$0.0031 & 38.89$\pm$0.43 & 0.9594$\pm$0.0033 \\
LEARN& 42.44$\pm$0.35 &0.9780$\pm$0.0021 &41.06$\pm$0.42 & 0.9738$\pm$0.0021 & 39.33$\pm$0.48 &0.9638$\pm$0.0024 \\
TGV & 42.53$\pm$0.52 & 0.9818$\pm$0.0017 & 41.44$\pm$0.32 &0.9762$\pm$0.0032 & 39.60$\pm$0.22 & 0.9601$\pm$0.0080 \\
LDA & 43.37$\pm$0.38 &0.9829$\pm$0.0017 & 41.64$\pm$0.47 &0.9763$\pm$0.0023 & 39.99$\pm$0.55 &0.9678$\pm$0.0021 \\
MAGIC & {44.58$\pm$0.37} & {0.9866$\pm$0.0017}& {43.10$\pm$0.28}& {0.9821$\pm$0.0024} & {41.12$\pm$0.20} & {0.9707$\pm$0.0058} \\
\textbf{ELDA (Proposed)} & \textbf{44.65$\pm$0.53} & \textbf{0.9867$\pm$0.0016} & \textbf{43.33$\pm$0.45} &\textbf{0.9826$\pm$0.0017} & \textbf{41.61$\pm$0.55}& \textbf{0.9763$\pm$0.0017}\\
\bottomrule
\end{tabular}
\label{table_result_2}
\end{table}
\section{Conclusion}
In brief, we propose an efficient inexact learned descent algorithm for low-dose CT reconstruction.
With incorporating the sparsity enhancing and non-local smoothing modules in the regularizer, the proposed ELDA outperforms several existing state-of-the-art reconstruction methods in accuracy and efficiency on two widely known datasets and retains convergence property.

\newpage
\bibliographystyle{IEEEtran}
\bibliography{ref}

\end{document}